\documentclass[letterpaper, UKenglish, cleveref, autoref, thm-restate, numberwithinsect]{lipics-v2021}

\pdfoutput=1 
\hideLIPIcs  

\usepackage[bibliography=common]{apxproof}

\theoremstyle{plain}
\newtheoremrep{claimapp}[theorem]{Claim}
\theoremstyle{plain}
\newtheoremrep{theoremapp}[theorem]{Theorem}
\theoremstyle{plain}
\newtheoremrep{lemmaapp}[theorem]{Lemma}

\title{Maximum Unique Coverage on Streams: Improved FPT Approximation Scheme and Tighter Space Lower Bound} 

\titlerunning{Max.\ Unique Coverage on Streams: Improved FPT-AS and Tighter Space Lower Bound} 

\author{Philip Cervenjak}{School of Computing and Information Systems, The University of Melbourne, Australia}{pcervenjak@student.unimelb.edu.au}{https://orcid.org/0000-0002-8349-619X}{This work was supported by a Elizabeth and Vernon Puzey Scholarship, and by the Faculty of Engineering and Information Technology.}

\author{Junhao Gan}{School of Computing and Information Systems, The University of Melbourne, Australia}{junhao.gan@unimelb.edu.au}{https://orcid.org/0000-0001-9101-1503}{}

\author{Seeun William Umboh}{School of Computing and Information Systems, The University of Melbourne, Australia \\ ARC Training Centre in Optimisation Technologies, Integrated Methodologies, and Applications (OPTIMA)}{william.umboh@unimelb.edu.au}{https://orcid.org/0000-0001-6984-4007}{}

\author{Anthony Wirth}{School of Computer Science, The University of Sydney, Australia \and School of Computing and Information Systems, The University of Melbourne, Australia}{anthony.wirth@sydney.edu.au}{https://orcid.org/0000-0003-3746-6704}{}

\authorrunning{P. Cervenjak, J. Gan, S. W. Umboh, A. Wirth} 

\Copyright{Philip Cervenjak, Junhao Gan, Seeun William Umboh, Anthony Wirth} 



\category{} 

\relatedversion{} 

\acknowledgements{We thank the anonymous reviewers for their valuable feedback.}

\nolinenumbers 

\DeclareMathOperator*{\argmin}{arg\,min}
\DeclareMathOperator{\poly}{poly}
\DeclareMathOperator{\polylog}{polylog}
\DeclareMathOperator*{\E}{\mathbb{E}}
\DeclareMathOperator{\Freq}{freq}

\let\oldnl\nl
\newcommand{\nonl}{\renewcommand{\nl}{\let\nl\oldnl}}

\newcommand{\Univ}{U}
\newcommand{\UnivSample}{\Univ'}
\newcommand{\UniqUniv}{\Tilde{\Univ}}
\newcommand{\Stream}{\mathcal{V}}
\newcommand{\StreamSample}{\Stream'}
\newcommand{\Part}{P}
\newcommand{\Elem}{x}

\newcommand{\UnivSize}{n}
\newcommand{\NumSets}{m}
\newcommand{\Pass}{p}
\newcommand{\SetS}{S}
\newcommand{\SetSSample}{\SetS'}
\newcommand{\SetT}{T}
\newcommand{\SetY}{Y}

\newcommand{\Indi}{i}
\newcommand{\SetSmall}{T}
\newcommand{\SetD}{D}

\newcommand{\Col}{\mathcal{C}}
\newcommand{\SubCol}{\mathcal{Q}}

\newcommand{\ColSize}{\ell}
\newcommand{\Sol}{\mathcal{B}}
\newcommand{\SolSample}{\Sol'}
\newcommand{\RecCol}{\Col\setminus\{\SetSmall\}}
\newcommand{\Cover}{\psi}
\newcommand{\UniqCover}{\Tilde{\Cover}}
\newcommand{\NUniqCover}{\Cover_{\geq 2}}
\newcommand{\Card}{k}
\newcommand{\MaxFreq}{r}
\newcommand{\MaxFreqCol}{r}
\newcommand{\MaxSize}{d}
\newcommand{\MaxSizeCol}{d}

\newcommand{\Error}{\varepsilon}
\newcommand{\SecError}{\Error'}

\newcommand{\Num}{z}
\newcommand{\Har}{H}
\newcommand{\HarNum}{\Num}
\newcommand{\HarInd}{t}
\newcommand{\OptValue}{\mathrm{OPT}}
\newcommand{\OptValueSample}{\OptValue'}
\newcommand{\Opt}{\mathcal{O}}
\newcommand{\OptIn}{\Opt^{\mathrm{in}}}
\newcommand{\OptOut}{\Opt^{\mathrm{out}}}
\newcommand{\OptSub}{\SubCol}
\newcommand{\StrInd}{i}
\newcommand{\IndStar}{\StrInd^*}
\newcommand{\PlayInd}{j}
\newcommand{\FreqInd}{t}
\newcommand{\FreqIndA}{u}
\newcommand{\SizeParam}{a}
\newcommand{\Space}{s}
\newcommand{\Prop}{q}
\newcommand{\NumSel}{\ell}
\newcommand{\SelCol}{\mathscr{L}}
\newcommand{\DistCol}{\SelCol_{\mathrm{di}}}
\newcommand{\IdenCol}{\SelCol_{\mathrm{id}}}
\newcommand{\RandVar}{X}
\newcommand{\Mean}{\mu}
\newcommand{\Bound}{b}
\newcommand{\Ratio}{\alpha}
\newcommand{\UniqCRatio}{\phi}

\newcommand{\Group}{\mathcal{G}}
\newcommand{\FreqCol}{\hat{\Col}}
\newcommand{\FreqSol}{\hat{\Sol}}

\newcommand{\FreqColSize}{\hat{\ColSize}}
\newcommand{\FreqColInd}{i}
\newcommand{\FreqColIndj}{j}

\newcommand{\SetOrder}{t}
\newcommand{\FreqError}{\Error_{\MaxFreqCol}}

\newcommand{\SizeCol}{\hat{\Col}}
\newcommand{\SizeSol}{\hat{\Sol}}
\newcommand{\FreqdUniv}{\hat{\Univ}}
\newcommand{\SizeError}{\Error_\MaxSizeCol}
\newcommand{\SecSizeError}{\hat{\Error}_\MaxSizeCol}

\newcommand{\TopSets}{\mathcal{A}}
\newcommand{\UniqRatio}{\phi}

\newcommand{\RandCol}{\mathcal{Z}}
\newcommand{\ProbSelSet}{p}

\newcommand{\OptGuess}{v}
\newcommand{\Const}{c}
\newcommand{\ConstA}{\Const_{1}}
\newcommand{\ConstB}{\Const_{2}}
\newcommand{\HashFunc}{h}
\newcommand{\ProbHash}{p}
\newcommand{\NumSetsStored}{s}

\newcommand{\Event}{\mathcal{E}}

\newcommand{\UG}{\textsc{UniqueGreedy}}

\newcommand{\UGF}{\textsc{UniqueGreedyFreq}}

\newcommand{\UGS}{\textsc{UniqueGreedySize}}

\newcommand{\UTS}{\textsc{UniqueTopSets}}

\newcommand{\RatioUGF}{\beta}
\newcommand{\Alg}{\textsc{Alg}}

\newcommand{\ParamFunc}{g}
\newcommand{\ParamList}{\boldsymbol{\gamma}}
\newcommand{\Inst}{\mathcal{I}}
\newcommand{\InstKernel}{\Inst'}

\newcommand{\MUCprob}{\textbf{Max Unique Coverage}}
\newcommand{\MCprob}{\textbf{Max Coverage}}
\newcommand{\UCprob}{Unique Coverage}

\newcommand{\Disj}{\textbf{Disj}}

\allowdisplaybreaks

\begin{document}

\begin{titlepage}
\def\thepage{}
\thispagestyle{empty}
\maketitle

\begin{abstract}
    We consider the Max Unique Coverage problem, including applications to the data stream model. The input is a universe of~$\UnivSize$ elements, a collection of~$\NumSets$ subsets of this universe, and a cardinality constraint,~$\Card$. The goal is to select a subcollection of at most~$\Card$ sets that maximizes unique coverage, i.e, the number of elements contained in exactly one of the selected sets. 
    The Max Unique Coverage problem has applications in wireless networks, radio broadcast, and envy-free pricing.

    Our first main result is a fixed-parameter tractable approximation scheme (FPT-AS) for Max Unique Coverage, parameterized by $\Card$ and the maximum element frequency,~$\MaxFreq$, which can be implemented on a data stream. Our FPT-AS finds a $(1-\Error)$-approximation while maintaining a kernel of size $\tilde{O}(\Card \MaxFreq/\Error)$, which can be combined with subsampling to use $\tilde{O}(\Card^2 \MaxFreq / \Error^3)$ space overall. This significantly improves on the previous-best FPT-AS with the same approximation, but a kernel of size $\tilde{O}(\Card^2 \MaxFreq / \Error^2)$. In order to achieve our result, we show upper bounds on the ratio of a collection's coverage to the unique coverage of a maximizing subcollection; this is by constructing explicit algorithms that find a subcollection with unique coverage at least a logarithmic ratio of the collection's coverage. We complement our algorithms with our second main result, showing that $\Omega(\NumSets / \Card^2)$ space is necessary to achieve a $(1.5 + o(1))/(\ln \Card - 1)$-approximation in the data stream. This dramatically improves the previous-best lower bound showing that $\Omega(\NumSets / \Card^2)$ is necessary to achieve better than a $e^{-1+1/\Card}$-approximation.
\end{abstract}
\end{titlepage}
\clearpage

\section{Introduction}
We study the \MUCprob{} problem, where we are given a universe of $\UnivSize$ elements, a collection of $\NumSets$ subsets of the universe, and an integer $\Card \in \{1, \dots, \NumSets\}$. The goal is to select a collection of at most~$\Card$ subsets that maximizes the number of elements covered by \emph{exactly} one set in the collection. This problem is a natural variant of the classic \MCprob{} problem, where the goal is to select a collection of $\Card$ subsets that maximizes the number of elements covered by \emph{at least} one set in the collection.

A weighted version of \MUCprob{} was first formally studied by Demaine et al.~\cite{Demaine2008}. In their motivating scenario, a number of wireless base stations, each with an associated cost, must be placed to maximize the number of mobile clients served. However, due to interference, if covered by more than one base station, a client receives bad service. Demaine et al.\ point out further applications to radio broadcast and envy-free pricing. They then showed an offline polynomial-time $\Omega(1 / \log \NumSets)$-approximation algorithm for their problem, which easily translates to a $\Omega( 1 / \log \Card)$-approximation for our problem.\footnote{This is by assuming that all sets have unit cost and that the budget is $\Card$.} Under various complexity assumptions, they showed (semi-)logarithmic inapproximability for polynomial-time algorithms; Guruswami and Lee \cite{Guruswami2017} later proved nearly logarithmic inapproximability, assuming NP does not admit quasipolynomial-time algorithms.

\subparagraph*{Streaming.}
Our work emphasizes solving \MUCprob{} approximately in the data stream model. All previous works, except McGregor et al.~\cite{mcgregor2021maximum}, only consider this problem in the offline model. In the data stream model, we focus on \emph{set-streaming}: each set in the stream is fully specified before the next; this setting is assumed in related works \cite{saha2009maximum,ausiello2012online,yu2013set,mcgregor2019better,mcgregor2021maximum}. We also constrain the space, measured in bits, to be $o(\NumSets \UnivSize)$, i.e., sublinear in both the number of sets,~$\NumSets$, and the size of the universe,~$\UnivSize$. Thus, we define the \MUCprob{} problem to include the cardinality constraint,~$\Card$. Previous works often formulate this problem without a cardinality constraint, simply referring to it as the `\UCprob{}' problem; this is equivalent to our formulation when $\Card = \NumSets$.

We are particularly interested in \MUCprob{} when parameterized by the \emph{maximum frequency},~$\MaxFreq$, defined as the maximum number of sets that an element belongs to; we also consider the \emph{maximum set size} $\MaxSize$ to a lesser extent. Parameter~$\MaxFreq$ has received considerable attention in studying fixed-parameter tractable (FPT) algorithms for classic coverage problems \cite{skowron2015fully,Bonnet2016,SKOWRON201765,Manurangsi2018,mcgregor2021maximum,Sellier2023}, but not as much for the \MUCprob{} problem \cite{mcgregor2021maximum}.

A central idea in achieving both FPT space and running time bounds is \emph{kernelization}. We transform a problem instance,~$\Inst$, into a smaller problem instance,~$\InstKernel$, called the \emph{(approximate) kernel}, such that $|\InstKernel| \leq \ParamFunc(\ParamList)$, where $\ParamFunc$ is a computable function in terms of problem parameters~$\ParamList$, while~$\InstKernel$ (approximately) preserves the optimal solution value of $\Inst$; a good solution can be found by brute-force search within $\InstKernel$. Consistent with parameterized streaming~\cite{Chitnis2015,Chitnis2016,chitnis2019towards,mcgregor2021maximum}, we further require an FPT streaming algorithm to use $O(\ParamFunc(\ParamList) \polylog|\Inst|)$ space.

\subsection{Our Contributions} \label{sec:contr}
Our first main result is a fixed-parameter tractable approximation scheme (FPT-AS) for \MUCprob{} with strong space, running time, and approximation bounds, that is applicable to the data stream model. A crucial step in achieving these performance bounds is showing improved upper bounds on what we call the \emph{unique coverage ratio} of a collection~$\Col$. This is the ratio between the coverage of~$\Col$ and the maximum unique coverage over all subcollections $\SubCol \subseteq \Col$; we let $\UniqRatio$ denote an \emph{upper bound} on the unique coverage ratio. We first outline the performance bounds of our FPT-AS in terms of $\UniqRatio$.

\subparagraph*{Main Result 1: FPT Approximation Scheme.}
We propose the FPT-AS \UTS{}, parameterized by the cardinality constraint $\Card$ and the maximum frequency $\MaxFreq$, which can be easily implemented in the data stream model. It achieves a $(1-\Error)$-approximation using a kernel of size $\left\lceil {\Card \MaxFreq (\UniqRatio + 1)}/{\Error} \right\rceil$. We formally present this algorithm in \cref{thm:UTS}.

\UTS{} is a refined version of the FPT-AS in Theorem 12 of McGregor et al.~\cite{mcgregor2021maximum}, in that our algorithm achieves a $(1-\Error)$ rather than a $(1/2-\Error)$-approximation using only an extra logarithmic factor of $(\UniqRatio+1)$ in the kernel size. Further, our algorithm improves on the FPT-AS in Theorem 10 of McGregor et al.~\cite{mcgregor2021maximum} by saving a factor of $O(\Card \log \NumSets / \Error)$ in the kernel size, and therefore a factor of $[O( \Card \log \NumSets / \Error )]^{\Card}$ in the running time, while achieving the same approximation factor. See \cref{table:FPT algorithms} for a comparison of our FPT-AS with others.
\begin{table}[htbp]
    \caption{Comparison of several FPT-AS for \MUCprob{}, parameterized by cardinality constraint,~$\Card$, and maximum frequency,~$\MaxFreq$. Note that the running time of each algorithm below is implied by its kernel size. Each finds a solution of size at most $\Card$ 
    by brute-force search in the kernel. Below, we can assign $\UniqRatio = \min( \ln \Card + 1, 2 \ln \MaxFreq + o(\log \MaxFreq), 2 \ln \MaxSize + o(\log \MaxSize))$.}
    \label{table:FPT algorithms}
    \begin{center}
    \begin{tabular}{l c c}
        \toprule
        \textbf{Reference} & \textbf{Approx.} & \textbf{Kernel Size}\\
        \midrule
         \cite[Theorem 10]{mcgregor2021maximum} & $1-\Error$ & $O\left({\Card^2 \MaxFreq \log^2 \NumSets}/{\Error^{2}} \right)$ \\
        \midrule
         \cite[Theorem 12]{mcgregor2021maximum} & ${1}/{2}-\Error$ & $\left\lceil {\Card \MaxFreq}/{\Error} \right\rceil$ \\
        \midrule
        {\cref{thm:UTS}, [ours,\UTS{}] } & $1-\Error$ & $\left\lceil {\Card \MaxFreq (\UniqRatio + 1)}/{\Error} \right\rceil$ \\
        \bottomrule
    \end{tabular}
    \end{center}
\end{table}

\subparagraph*{Unique Coverage Algorithms.}
In order to show good values for $\UniqRatio$, we propose a number of offline polynomial-time algorithms that, given an arbitrary $\Col$, explicitly return a $\Sol \subseteq \Col$ whose unique coverage is at least a logarithmic ratio of $\Col$'s coverage. We refer to them as \emph{unique coverage algorithms}; in fact, they can be thought as approximation algorithms for the unconstrained \UCprob{} problem on an input instance of $\Col$.

Our three offline polynomial-time algorithms, \UG{}, \UGF{}, and \UGS{}, each take a collection of sets,~$\Col$, and return a collection,~$\Sol \subseteq \Col$, whose unique coverage is at least a $1/(\ln \ColSize + 1)$, $1/(2 \ln \MaxFreq + o(\log \MaxFreq))$, and $1/(2 \ln \MaxSize + o(\log \MaxSize))$ proportion of $\Col$'s coverage respectively; in this context, $\ColSize = |\Col|$, $\MaxFreq$ is the maximum frequency in $\Col$, and $\MaxSize$ is the maximum set size in $\Col$. We formally present these algorithms in
\cref{thm:UG ratio}, \cref{thm:UGF ratio}, and in \cref{thm:UGS ratio},
respectively. See \cref{table:poly-time algorithms} for a comparison of our algorithms with those of Demaine et al. \cite{Demaine2008} along with their implied bounds,~$\UniqRatio$, albeit weaker than ours.

\subparagraph{Implication for FPT Approximation Scheme.}
The bound on the unique coverage ratio,~$\UniqRatio$, affects the kernel size and therefore the brute-force running time of \UTS{}. In particular, when $\MaxFreq = \Omega(\sqrt{\Card})$, the bound of $\UniqRatio$ implied by \UG{} is $10.66$ times smaller than implied by Demaine et al.~\cite{Demaine2008}; whereas when $\MaxFreq = o(\sqrt{\Card})$, the bound of $\UniqRatio$ implied by \UGF{} is almost $5.33$ smaller than implied by Demaine et al. This means, by using our implied bounds rather than those implied by Demaine et al., we save a factor of $10.66^{\Card}$ in \UTS{}'s running-time when $\MaxFreq = \Omega(\sqrt{\Card})$, and a factor of almost $5.33^{\Card}$ when $\MaxFreq = o(\sqrt{\Card})$.

\subparagraph*{Improvements in Polynomial-Time Approximation.}
As a separate contribution, each of our three unique coverage algorithms finds a logarithmic approximation to \MUCprob{}, both offline and in the data stream. We first find a solution $\Col$ to \MCprob{} in polynomial time, and then run one of our above algorithms on $\Col$ to return the subcollection $\Sol \subseteq \Col$. For this purpose, our algorithms \UG{}, \UGF{}, and \UGS{} improve the approximation factor due to Demaine et al.~\cite{Demaine2008} by a factor of~$10.66$,~$5.33$, and~$10.66$, respectively. Following the above approach, we propose a single-pass streaming algorithm for \MUCprob{}  that achieves a $(1/(2\UniqRatio) - \Error)$-approximation using $\tilde{O}({\Card^2}/{\Error^3})$ space, where we can assign $\UniqRatio = \min( \ln \Card + 1, 2 \ln \MaxFreq + o(\log \MaxFreq), 2 \ln \MaxSize + o(\log \MaxSize))$. We formally state this in \cref{thm:poly time data stream}.
\begin{table}[h]
    \caption{Polynomial-time algorithms for \MUCprob{}. Compared to others, our methods imply constant-factor improvements in the unique coverage ratio bound, $\UniqRatio$.}
    \label{table:poly-time algorithms}
    \begin{center}
    \begin{tabular}{l l c}
        \toprule
        \textbf{Parameter} & \textbf{Reference} & \textbf{(Implied) $\UniqRatio$} \\
        \midrule
        \multirow{2}{*}[-0.7em]{$\ColSize = \text{collection size}$} &
        \cite[Theorem 4.1]{Demaine2008} & $10.66 \ln (\ColSize+1)$ \\
        \cmidrule{2-3}
        & \makecell[l]{Ours, \cref{thm:UG ratio} \\(\UG{})} & $\ln \ColSize + 1$ \\
        \midrule
        \multirow{2}{*}[-0.7em]{\makecell[l]{$\MaxFreqCol = \text{maximum frequency}$ \\in a collection }} & \cite[Theorem 4.1]{Demaine2008} & $10.66 \ln (\MaxFreqCol+1)$ \\
        \cmidrule{2-3}
        & \makecell[l]{Ours, \cref{thm:UGF ratio} \\(\UGF{})} & $2\ln \MaxFreqCol + o(\log \MaxFreqCol)$ \\
        \midrule
        \multirow{2}{*}[-0.7em]{\makecell[l]{$\MaxSizeCol = \text{maximum set size}$ \\in a collection}} & \cite[Theorem 4.2]{Demaine2008} & $21.32 \ln (\MaxSizeCol+1)$ \\
        \cmidrule{2-3}
        & \makecell[l]{Ours, \cref{thm:UGS ratio} \\(\UGS{})} & $ 2 \ln \MaxSizeCol + o(\log \MaxSizeCol) $ \\
        \bottomrule
    \end{tabular}
    \end{center}
\end{table}

\subparagraph*{Main Result 2: Streaming Lower Bound.}
Our second main result is a significantly  improved streaming lower bound for \MUCprob{}.
In the data stream model, we prove that any randomized algorithm that achieves a $(1.5 + o(1))/(\ln \Card - 1)$-approximation for \MUCprob{} w.h.p.\ requires $\Omega({\NumSets}/{\Card^2})$ space. We formally state this in \cref{thm:lower bound adv}. Our lower bound improves on the lower bound by McGregor et al. \cite{mcgregor2021maximum}, which shows a similar result, but achieves w.h.p.~a $e^{-1+1/\Card} \geq 1/e$-approximation.  Interestingly, our approximation threshold is close to~$3$ times larger than the approximation (in terms of $\Card$) achieved by our $\tilde{O}(\Card^2/\Error^3)$ space algorithm in \cref{thm:poly time data stream}, indicating that a dramatic increase in space is needed to bridge this approximation gap.

\subsection{Technical Overview} \label{sec:technical overview}

\subparagraph*{FPT Approximation Scheme.}
\UTS{} refines the technique used in the FPT-AS for \MUCprob{} in Theorem 12 of McGregor et al. \cite{mcgregor2021maximum}, which is to construct an approximate kernel by storing a number of the largest sets by individual size, and then to find a subcollection of the kernel with maximum unique coverage by brute-force search. Similar techniques have been used in FPT-AS approaches for Max Vertex Cover~\cite{Manurangsi2018,Huang2022} and \MCprob{}~\cite{skowron2015fully,SKOWRON201765,mcgregor2021maximum,Sellier2023}. Our novelty is providing a stronger analysis of the approximation factor preserved by the kernel, allowing us to achieve a $(1-\Error)$-approximation while only increasing the kernel size by a logarithmic factor in $\Card$, $\MaxFreq$, or $\MaxSize$.

\subparagraph*{Unique Coverage Algorithms.}
All of our unique coverage algorithms are combinatorial in design. Our first two, \UG{} and \UGF{}, are novel algorithms that each, in some sense, use a greedy approach, noting that \UG{} is used as subroutine of \UGF{}. Our third algorithm, \UGS{}, is easily derived by combining \UGF{} with the approach by Demaine et al. \cite{Demaine2008} for sets with maximum cost $\MaxSize$ (maximum size in our case).

\subparagraph*{Streaming Lower Bound.}
Our streaming lower bound relies on a novel reduction from $\Card$-player Set Disjointness in the one-way communication model to \MUCprob{} in the data stream. 
In the hard instance of \MUCprob{} thus constructed, either all collections of $\NumSel \leq \Card$ sets have a unique coverage of $\SizeParam \Card^2(1.5 + o(1))$ w.h.p.\ or there exists a single collection of $\Card$ sets whose unique coverage is at least $\SizeParam \Card^2(\ln \Card-1)$, where $\SizeParam = \Omega(\Card \log \NumSets)$. By a standard argument, we show that distinguishing between these instances of \MUCprob{} with a streaming algorithm is as hard as solving Set Disjointness, implying the required space lower bound.
\begin{table}[htbp]
    \caption{Comparison of space lower bounds for \MUCprob{} in the data stream. Note that the lower bound by Assadi \cite{assadi2017tight} was shown for \MCprob{} with constant $\Card = 2$, but it is not difficult to adapt it for \MUCprob{} because, in the hard instance constructed for the lower bound, the unique coverage of any pair of sets behaves similarly to its coverage.}
    \label{table:lower bounds}
    \begin{center}
    \begin{tabular}{l c c}
        \toprule
        \textbf{Reference} & \textbf{Approx.} & \textbf{Space LB} \\
        \midrule
        \cite[Theorem 4]{assadi2017tight} & $1-\Error$ & $\Omega\left({\NumSets}/{\Error^2}\right)$ \\
        \cmidrule{1-3}
        \cite[Theorem 16]{mcgregor2021maximum} & $1/e$ & $\Omega\left( {\NumSets}/{\Card^2} \right)$ \\
        \cmidrule{1-3}
        Ours, \cref{thm:lower bound adv} & $( 1.5 + o(1) )/(\ln\Card - 1)$ & $\Omega\left( {\NumSets}/{\Card^2} \right)$ \\
        \bottomrule
    \end{tabular}
    \end{center}
\end{table}

\subsection{Paper Structure}
After preliminaries in \cref{sec:prelim},  \cref{sec:FPT} presents our FPT-AS \UTS{} and a polynomial-time algorithm, both applicable to the data stream. In \cref{sec:greedy algorithms}, we present our component algorithms for bounding the unique coverage ratio. In \cref{sec:lower bound}, we present a space lower bound for achieving a $(1.5 + o(1))/(\ln \Card - 1)$-approximation for \MUCprob{}. We conclude in \cref{sec:conclusions}.

\section{Preliminaries} \label{sec:prelim}
\subparagraph{Notation.}
For convenience, we hence let~$[\UnivSize]$ denote the set of integers~$\{1,2,\ldots,\UnivSize\}$. Likewise, $\Univ = [\UnivSize]$ denotes a universe of $\UnivSize$ elements, while~$\Stream$ denotes a collection of~$\NumSets$ subsets of~$\Univ$.

Given a collection~$\Col$ of sets, the \emph{unique cover} of $\Col$ is the subset the universe covered by exactly one set in $\Col$.
Formally, $\UniqCover(\Col) \coloneqq (\bigcup_{\SetS \in \Col} \SetS) \setminus (\bigcup_{\SetS \neq \SetT \in \Col} \SetS \cap \SetT)$, and the \emph{unique coverage} of $\Col$ is~$|\UniqCover(\Col)|$.
For convenience, the \emph{cover} of~$\Col$ is the union of the sets in~$\Cover(\Col) \coloneqq \bigcup_{\SetS \in \Col} \SetS$, and the \emph{coverage} of $\Col$ is $|\Cover(\Col)|$. Further, the \emph{non-unique cover} of~$\Col$ is the subset of the universe covered by at least two sets from $\Col$, denoted by $\NUniqCover(\Col) = \bigcup_{\SetS \neq \SetT \in \Col} \SetS \cap \SetT$ -- equivalently $\NUniqCover(\Col) = \Cover(\Col) \setminus \UniqCover(\Col)$ -- and the \emph{non-unique coverage} of~$\Col$ is $|\NUniqCover(\Col)|$. The \emph{maximum unique coverage} of~$\Col$ is the largest unique coverage of a sub-collection consisting of at most $k$ sets.\footnote{Unlike Max~Coverage, the optimal solution to Max Unique Coverage may contain fewer than $k$ sets.} The \emph{unique coverage ratio} of~$\Col$ is the ratio between its coverage and maximum unique coverage. In other words, if $\SubCol$ is the subcollection of $\Col$ that has maximum unique coverage, then the unique coverage ratio of $\Col$ is $|\Cover(\Col)|/|\UniqCover(\SubCol)|$.

Given an element $\Elem \in \Univ$ and a collection~$\Col$ of sets, the \emph{frequency} of $\Elem$ in $\Col$ is defined as $\Freq_{\Col}(\Elem) \coloneqq |\{ \SetS \in \Col : \Elem \in \SetS \}|$, i.e., the number of sets in $\Col$ that contain $\Elem$; and the \emph{maximum frequency} is defined as $\MaxFreq \coloneqq \max_{\Elem \in \Univ} \Freq_{\Col}(\Elem)$. Also, the \emph{maximum set size} is defined as $\MaxSize \coloneqq \max_{\SetS \in \Col} |\SetS|$. We often use $\MaxFreq$ and $\MaxSize$ to refer to the maximum frequency and set size, respectively, in $\Col = \Stream$ unless stated otherwise. Note that $\MaxFreq \leq |\Col|$ holds for every $\Col$.
We let $\Har_\HarNum \coloneqq \sum_{\HarInd=1}^{\HarNum}1/\HarInd$ denote the $\HarNum^{\text{th}}$ harmonic number, a
term that appears several times.

\subparagraph{Formal Problem Definition.}
An instance of \MUCprob{} consists of an element universe $\Univ$, a collection $\Stream$ of $m$ subsets of $\Univ$, and an integer $\Card \in [\NumSets]$; when the context is clear, we represent an instance with just $\Stream$ for simplicity. 
The goal of \MUCprob{} is to return a subcollection $\Sol \subseteq \Stream$ (more precisely, a collection of IDs of sets), with $|\Sol| \leq \Card$, that maximizes $|\UniqCover(\Sol)|$. We let~$\Opt$ denote an optimal solution to this \MUCprob{} problem, and $\OptValue \coloneqq |\UniqCover(\Opt)|$ as the maximum unique coverage.

\subparagraph{Subsampling for the Data Stream Model.}
The universe subsampling technique has been widely successful in the development of streaming algorithms for coverage problems \cite{DemaineIMV14,Har-PeledIMV16,BateniEM17,mcgregor2019better}.
In this work, we follow the approach of McGregor and Vu~\cite{mcgregor2019better}, and sample the universe so that each set has size~$O(\Card \log \NumSets / \Error^2)$.
We assume that $\Card \in o(\NumSets \UnivSize)$, and also that $\Card$ is known prior to reading the stream.
The main result is given in the following lemma, with
a proof sketch of the subsampling approach in \cref{sec:subsampling approach}.

\begin{lemma}[Subsampling Approach \cite{mcgregor2021maximum}] \label{lemma:subsampling approach}
    Let $\Error \in (0,1)$ be the subsampling error parameter. Given an instance of \MUCprob{} and an $\Ratio$-approximation streaming algorithm, we can run the algorithm on $\lceil \log_2 \UnivSize \rceil$ parallel subsampled instances and select one of them such that the algorithm's solution corresponds to a $(\Ratio-2\Error)$-approximation for the original instance with probability $1-1/\poly(\NumSets)$. Moreover, if the streaming algorithm stores at most $\NumSetsStored$ sets in every subsampled instance, then the total space complexity of the subsampling approach is bounded by $\lceil \log_2 \UnivSize \rceil \cdot \NumSetsStored \cdot
    O\left(\Card \log \NumSets \log \UnivSize/{\Error^2}\right)$.
\end{lemma}

\section{Streaming FPT-AS and Polynomial-Time Algorithms}
\label{sec:FPT}

In Section~\ref{sec:UTS Lemma}, we prove a kernelization lemma. Then, we use it to obtain an FPT-AS and a parameterized streaming algorithm in Section~\ref{sec:UTS}. Finally, we show how to use a bound on the unique coverage ratio to obtain a polynomial-time streaming algorithm in Section~\ref{sec:app of poly time to streams}.

\subsection{Kernelization Lemma} 
\label{sec:UTS Lemma}
Our Kernelization Lemma below, as well as its proof, is a refinement of Lemma~11 by McGregor et al.~\cite{mcgregor2021maximum}. We first provide some intuition on why our kernel preserves a $(1-\Error)$-approximation for \MUCprob{}.

\subparagraph*{Intuition of Kernelization Lemma.}
For convenience, let $\SecError$ be an intermediate error parameter and define the kernel $\TopSets$ as the collection of $\lceil \Card \MaxFreq/\SecError \rceil$ largest sets in instance $\Stream$ by individual size. Given the optimal solution for \MUCprob{}, $\Opt$, let~$\OptIn$ and~$\OptOut$ be the collections of optimal sets found and not found in $\TopSets$ respectively.

One main step in proving our Kernelization Lemma is showing that, in expectation, a collection of $|\OptOut|$ sets sampled without replacement from $\TopSets$, denoted by  $\RandCol$, can be appended to $\OptIn$ with little overlap in their unique covers. In particular, we can prove that $\E[|\UniqCover(\OptIn \cup \RandCol)|] \geq (1-\SecError)|\UniqCover(\Opt)| - \SecError|\Cover(\Opt)|$. 

However, due to the~$\SecError|\Cover(\Opt)|$ term, this is not enough to achieve the required approximation factor. This term reflects the fact that, even if the unique cover of~$\OptIn$ has little overlap with the unique cover of~$\RandCol$, the \emph{entire} cover of $\OptIn$ could be more extensive and, thus, overlap significantly with the unique cover of $\RandCol$. To address this, in \cref{claim:opt unique cover ratio}, we show $\UniqRatio |\UniqCover(\Opt)| \geq |\Cover(\Opt)|$, where~$\UniqRatio$ upper bounds the unique coverage ratio. Substituting this into the lower bound for $\E[|\UniqCover(\OptIn \cup \RandCol)|]$, and assigning $\SecError = \Error/(\UniqRatio + 1)$, we obtain $\E[|\UniqCover(\OptIn \cup \RandCol)|] \geq (1-\Error)|\UniqCover(\Opt)|$, implying the existence of a $(1-\Error)$-approximate subcollection of $\TopSets$. Lastly, the final kernel size of $|\TopSets| = \lceil \Card \MaxFreq (\UniqRatio + 1)/\Error \rceil$ follows from the assignment of $\SecError$.

\begin{lemma}[Kernelization Lemma]
    \label{lemma:kernel}
    Suppose that every collection of sets has unique coverage ratio at most $\UniqCRatio$.
    Let $\Stream$ denote a collection of sets. Then, for every $\Error \in (0,1)$, the sub-collection,~$\TopSets$, of the $\lceil \Card \MaxFreq (\UniqRatio + 1)/\Error \rceil$-largest sets of $\Stream$ (by size) contains a sub-collection of at most $k$ sets with unique coverage at least $(1-\Error)\OptValue$.
\end{lemma}
    \begin{proof}    
    Assume that $|\Stream| \geq \lceil \Card \MaxFreq (\UniqRatio + 1)/\Error \rceil$: otherwise,~$\TopSets$ would contain every set in~$\Stream$ and so would trivially have~$\Opt$ as a subcollection. Let $\OptIn = \Opt \cap \TopSets$ and $\OptOut = \Opt \setminus \TopSets$. Let~$\RandCol$ be a uniform random sample of~$|\OptOut|$ sets chosen from~$\TopSets$ without replacement. The main goal is to prove \cref{claim:ex uniq cover col}, below. Since $\OptIn$ and $\RandCol$ are subsets of $\TopSets$, this implies the existence of subcollection $\Sol \subseteq \TopSets$ as required by the lemma.
  
    We start with the following lower bound on the expected unique coverage of $\OptIn \cup \RandCol$, as shown in inequality~\eqref{eqn:ex uniq cover opt in rand col} below. Then we lower bound each of the RHS terms separately and simplify afterwards. By definition,
    \begin{align}
        |\UniqCover(\OptIn \cup \RandCol)| &\geq |\UniqCover(\OptIn)| + |\UniqCover(\RandCol)| - (|\UniqCover(\OptIn) \cap \Cover(\RandCol)| + |\Cover(\OptIn) \cap \UniqCover(\RandCol)|)\,, \notag \\ \shortintertext{hence, by linearity of expectation,}
        \E[ |\UniqCover(\OptIn \cup \RandCol)| ] &\geq |\UniqCover(\OptIn)| + \E[ |\UniqCover(\RandCol)| ] - \E[|\UniqCover(\OptIn) \cap \Cover(\RandCol)|] - \E[|\Cover(\OptIn) \cap \UniqCover(\RandCol)|]\,.
        \label{eqn:ex uniq cover opt in rand col}
    \end{align}
    
    Define an intermediate error parameter, $\SecError={\Error}/(\UniqRatio + 1)$, meaning $|\TopSets| = \lceil \Card \MaxFreq / \SecError \rceil$. The probability of a set $\SetS \in \TopSets$ being selected in $\RandCol$ is
$  
        \ProbSelSet \coloneqq |\OptOut|/|\TopSets| \leq \Card/(\Card \MaxFreq / \SecError) = \SecError/\MaxFreq
$.
    Now \cref{claim:exp uniq cov rand col}, below, is easily derived from the proof of Lemma~11 in by McGregor et al.~\cite{mcgregor2021maximum}.
    \begin{claim} \label{claim:exp uniq cov rand col}
        It holds that $\E[ |\UniqCover(\RandCol)| ] \geq (1-\SecError) |\UniqCover(\OptOut)|$\,.
    \end{claim}
    \begin{proof}
    Quantity~$|\UniqCover(\RandCol)|$ can be lower bounded by summing, over every $\SetS \in \RandCol$, the number of elements in~$\SetS$ not contained in any other~$\SetT \in \RandCol \setminus \{\SetS\}$. From there, we prove inequality~\eqref{claim:exp uniq cov rand col}, below. We let~$[ \Event ]$ denote the indicator variable for event $\Event$.
    \begin{align}
        &|\UniqCover(\RandCol)| \geq \sum_{\SetS \in \RandCol} \left(|\SetS| - \sum_{\SetT \in \RandCol \setminus \{\SetS\}} |\SetS \cap \SetT| \right), \text{\ \ hence,}\notag  \\
        &\E[|\UniqCover(\RandCol)|] \notag \\
        &\geq \E \left[ \sum_{\SetS \in \TopSets} \left(|\SetS|[\SetS \in \RandCol] - \sum_{\SetT \in \TopSets \setminus \{\SetS\}} |\SetS \cap \SetT|[\SetS \in \RandCol \wedge \SetT \in \RandCol] \right) \right] \notag \\
        &\geq \sum_{\SetS \in \TopSets} \left( |\SetS| \ProbSelSet - \sum_{\SetT \in \TopSets \setminus \{\SetS\}} |\SetS \cap \SetT|\ProbSelSet^2 \right) &\text{$\Pr[\SetS \in \RandCol \wedge \SetT \in \RandCol] \leq \ProbSelSet^2$} \notag \\
        &\geq \sum_{\SetS \in \TopSets} \left( |\SetS| \ProbSelSet - |\SetS| \ProbSelSet^2 (\MaxFreq-1) \right) &\begin{aligned}
        &&\text{each $\Elem \in \SetS$ intersects}\\
        &&\text{$\leq \MaxFreq-1$ other sets}
        \end{aligned} \notag \\
        &\geq \ProbSelSet(1-\ProbSelSet\MaxFreq) \sum_{\SetS \in \TopSets} |\SetS| \notag \\
        &\geq \ProbSelSet (1-\SecError) \sum_{\SetS \in \TopSets} |\SetS| &\text{$\ProbSelSet \leq \frac{\SecError}{\MaxFreq}$} \notag \\
        &\geq \ProbSelSet (1-\SecError) |\TopSets| \frac{\sum_{\SetY \in \OptOut} |\SetY|}{|\OptOut|} \geq \ProbSelSet (1-\SecError) |\TopSets| \frac{|\UniqCover(\OptOut)|}{|\OptOut|} &\begin{aligned}
        &&\text{for all $\SetS \in \TopSets$ and all}\\
        &&\text{$\SetY \in \OptOut \colon |\SetS| \geq |\SetY|$}
        \end{aligned} \notag \\
        &= \ProbSelSet (1-\SecError) \frac{|\UniqCover(\OptOut)|}{\ProbSelSet} \notag = (1-\SecError)|\UniqCover(\OptOut)|\,. \tag*{\qedhere}
    \end{align}
    \end{proof}
    
    \cref{claim:exp inter opt in}  upper bounds the expected size of the overlap between $\UniqCover(\OptIn)$ and $\Cover(\RandCol)$ and the expected size of the overlap between $\UniqCover(\RandCol)$ and $\Cover(\OptIn)$.
    \begin{claim} \label{claim:exp inter opt in}
         $\E[ |\UniqCover(\OptIn) \cap \Cover(\RandCol)| ] \leq \SecError |\UniqCover(\OptIn)|$ and $\E[ |\Cover(\OptIn) \cap \UniqCover(\RandCol)| ] \leq \SecError |\Cover(\OptIn)|$.
    \end{claim}
    \begin{proof}
    To prove the first inequality,
    \begin{align}
        \E[ |\UniqCover(\OptIn) \cap \Cover(\RandCol)| ] \leq \sum_{\Elem \in \UniqCover(\OptIn)} \sum_{\SetS \in \TopSets \colon \Elem \in \SetS} \Pr[\SetS \in \RandCol] \leq \sum_{\Elem \in \UniqCover(\OptIn)} \MaxFreq \ProbSelSet \leq \SecError |\UniqCover(\OptIn)|\,. \notag
    \end{align}
    To prove the second inequality,
    it is clear that $\UniqCover(\RandCol) \subseteq \Cover(\RandCol)$  for all~$\RandCol$,
    so we have $\E[ |\Cover(\OptIn) \cap \UniqCover(\RandCol)| ] \leq \E[ |\Cover(\OptIn) \cap \Cover(\RandCol)| ]$.
    Then
    substituting~$\Cover(\OptIn)$
    for~$\UniqCover(\OptIn)$
    in the argument for the first inequality, we see
    $\E[ |\Cover(\OptIn) \cap \Cover(\RandCol)| ] \leq \SecError |\Cover(\OptIn)|$.
    \end{proof}
We now turn to a property of the optimal solution for \MUCprob{},~$\Opt$.
    \begin{claim}
        \label{claim:opt unique cover ratio}
        $\UniqCRatio|\UniqCover(\Opt)| \geq |\Cover(\Opt)|$\,.
    \end{claim}
    \begin{proof}
        Recall that we assumed that every collection of sets has unique coverage ratio at most~$\UniqCRatio$. In particular, $\Opt$ has a sub-collection,~$\OptSub$, of at most~$k$ sets with $\UniqCRatio |\UniqCover(\OptSub)| \geq |\Cover(\Opt)|$.
        By optimality, $\Opt$'s unique coverage is at least that of $\OptSub$. Thus, we get the desired inequality.
    \end{proof}
Starting from Ineq.~\eqref{eqn:ex uniq cover opt in rand col},
    we can now lower bound $\E[ |\UniqCover(\OptIn \cup \RandCol)| ]$. 
    \begin{claim}
    We have the lower bound $\E[ |\UniqCover(\OptIn \cup \RandCol)| ] \geq (1-\Error) |\UniqCover(\Opt)|\,.$ \qedhere
    \label{claim:ex uniq cover col}
    \end{claim}
    \begin{proof}
    \begin{align}
        &\E[ |\UniqCover(\OptIn \cup \RandCol)| ] \notag \\
        &\geq |\UniqCover(\OptIn)| + \E[ |\UniqCover(\RandCol)| ] - \E[|\UniqCover(\OptIn) \cap \Cover(\RandCol)|] - \E[|\Cover(\OptIn) \cap \UniqCover(\RandCol)|] \notag &\text{Ineq.~\eqref{eqn:ex uniq cover opt in rand col}}\\
        &\geq |\UniqCover(\OptIn)| + (1-\SecError)|\UniqCover(\OptOut)| - \SecError |\UniqCover(\OptIn)| - \SecError |\Cover(\OptIn)| &\text{Claims~\ref{claim:exp uniq cov rand col} and \ref{claim:exp inter opt in}} \notag \\
        &= (1-\SecError) \left( |\UniqCover(\OptIn)| + |\UniqCover(\OptOut)| \right) - \SecError |\Cover(\OptIn)| \notag \\
        &\geq (1-\SecError)|\UniqCover(\Opt)| - \SecError |\Cover(\OptIn)| &\text{subadditivity of $\UniqCover$} \notag \\
        &\geq (1-\SecError)|\UniqCover(\Opt)| - \SecError |\Cover(\Opt)| &\text{monotonicity of $\Cover$} \notag \\
        &\geq (1-\SecError)|\UniqCover(\Opt)| -\SecError \UniqRatio |\UniqCover(\Opt)| &\text{\cref{claim:opt unique cover ratio}} \notag \\
        &= \left(1-\SecError\left(1 + \UniqRatio \right) \right)|\UniqCover(\Opt)| \notag \\
        &=(1-\Error) |\UniqCover(\Opt)|\,. &\text{$\SecError = \frac{\Error}{\UniqRatio + 1}$} \notag
    \end{align}
    \end{proof}
    \end{proof}

\subsection{Applications of the Kernelization Lemma}
\label{sec:UTS}
We now apply the Kernelization Lemma to prove the following theorem. 
\begin{theorem}
    \label{thm:UTS}
    Suppose that every collection of sets has unique coverage ratio at most $\UniqCRatio$.
    Let $\Stream$ denote a collection of sets, $\Card \geq 2$ denote the cardinality constraint, $\MaxFreq \geq 2$ denote the maximum frequency in $\Stream$, and $\Error \in (0,1)$ denote an error parameter. Then, there exist
    \begin{enumerate}
    \item an FPT-AS that finds a $(1-\Error)$-approximation for \MUCprob{} and has a running time of
    $\left( e \MaxFreq (\UniqRatio + 1)/{\Error} \right)^{\Card} \poly(\NumSets, \UnivSize, 1/\Error)$; and
    
    \item a streaming algorithm that finds a $(1-\Error)$-approximation for \MUCprob{} with probability $1 - 1/\poly(\NumSets)$ and uses $\tilde{O}(\UniqCRatio \Card^2 \MaxFreq / \Error^3)$ space.

    \end{enumerate}
\end{theorem}

Our algorithm, \UTS{}, takes a collection of sets $\Stream$ with maximum frequency $\MaxFreq \geq 2$, a cardinality constraint $\Card$, and an error parameter $\Error \in (0,1)$, and returns a $(1-\Error)$-approximation for \MUCprob{}. It also takes parameter~$\UniqCRatio$, an upper bound on the unique coverage ratio of every subcollection. $\UTS{}$ first finds $\TopSets$, the $\left\lceil \Card \MaxFreq (\UniqRatio + 1)/\Error \right\rceil$-largest sets $\SetS \in \Stream$ by size $|\SetS|$. Then, it brute-forces over $\TopSets$, i.e.~it finds the sub-collection of $\TopSets$ containing at most $k$ sets and has the maximum unique coverage.

\subparagraph{FPT-AS.} Let us first see how \UTS{} has the properties of the FPT-AS claimed in Theorem~\ref{thm:UTS}. The Kernelization Lemma (Lemma~\ref{lemma:kernel}) immediately implies that the solution returned by \UTS{} is a $(1-\Error)$-approximation. The running time bound follows by bounding the number of sub-collections of $\TopSets$ containing at most~$\Card$ sets.
\begin{lemma}
    \UTS{} has running time in $\left( e \MaxFreq (\UniqRatio + 1)/{\Error} \right)^{\Card} \poly(\NumSets, \UnivSize, 1/\Error)$.
\end{lemma}
\begin{proof}
\UTS{} considers every possible collection of $\NumSel \in [\Card]$ sets from $\TopSets$ and outputs the one with the best unique coverage. Below, the second inequality holds since replacing $\NumSel$ with $\Card$ makes each binomial coefficient larger, as $\NumSel \leq \Card \leq \Card \MaxFreq/2$ due to $\MaxFreq \geq 2$; the equality holds since $\binom{\Num+1}{\Card} = \binom{\Num}{\Card} (\Num+1)/(\Num + 1 - \Card)$; and the final inequality holds since $\binom{\Num}{\Card} \leq ( e \Num 
/ \Card )^{\Card}$. Thus, the running time is bounded as follows.
    \begin{align}
        &\sum_{\NumSel = 1}^{\Card} \binom{|\TopSets|}{\NumSel} \poly(\NumSets, \UnivSize) 
        \leq \poly(\NumSets, \UnivSize) \sum_{\NumSel = 1}^{\Card} \binom{\frac{\Card \MaxFreq (\UniqRatio + 1)}{\Error}+1}{\NumSel} 
        \leq \poly(\NumSets, \UnivSize) \Card \binom{\frac{\Card \MaxFreq (\UniqRatio + 1)}{\Error}+1}{\Card} \notag \\
        &= \poly(\NumSets, \UnivSize, 1/\Error) \binom{\frac{\Card \MaxFreq (\UniqRatio + 1)}{\Error}}{\Card}
        \leq \poly(\NumSets, \UnivSize, 1/\Error)\left( \frac{e \MaxFreq (\UniqRatio + 1)}{\Error} \right)^{\Card}. \tag*{\qedhere}
    \end{align}
\end{proof}
\subparagraph{Streaming Algorithm.}
This algorithm can also be run on a data stream using the subsampling approach in \cref{lemma:subsampling approach}, in which case it returns  w.h.p.\ a $(1-3\Error)$-approximation for \MUCprob{}.
\proofsubparagraph{Subsampled Streaming Approximation Factor.}
    Given an instance of \MUCprob{}, we can use the subsampling approach from \cref{sec:subsampling approach} and run \UTS{} on the subsampled instances in parallel (we use the same $\Error$ in the subsampling as in the algorithm). Since \UTS{} achieves a $(1-\Error)$-approximation for each subsampled instance by \cref{lemma:kernel}, it follows from \cref{lemma:subsampling approach} that we can return a $(1-3\Error)$-approximation for the original instance with probability $1-1/\poly(\NumSets)$.

    \proofsubparagraph{Subsampled Streaming Space Complexity.}
    Using the subsampling approach from \cref{lemma:subsampling approach}, \UTS{} stores $|\TopSets|  \leq \lceil \Card \MaxFreq (\UniqRatio + 1) / \Error \rceil $ sets in each subsampled instance. It follows from \cref{lemma:subsampling approach} that the overall space complexity is $\lceil\log_2 \UnivSize \rceil \cdot |\TopSets| \cdot \tilde{O}(\Card / \Error^2) = \tilde{O} ( \UniqRatio \Card^2 \MaxFreq / \Error^3 )$.
    
\subsection{Polynomial-Time Streaming Algorithm} \label{sec:app of poly time to streams}
Here we present a single-pass streaming algorithm that returns a $(1/(2\UniqRatio) - 3\Error)$-approximation for \MUCprob{}, given a bound on the unique coverage ratio, $\UniqRatio$. We present the algorithm in \cref{thm:poly time data stream} below. 

In the theorem statement, we assume we can use an offline polynomial-time algorithm, \Alg{}, that takes a collection $\Col$ and returns a subcollection $\Sol \subseteq \Col$ such that $|\UniqCover(\Sol)| \geq |\Cover(\Col)| / \UniqRatio$ for a ratio $\UniqRatio$ depending on $\Card = |\Col|$, the maximum frequency $\MaxFreq$, and the maximum set size~$\MaxSize$. \Alg{} can be substituted with a procedure that runs all of our unique coverage algorithms from \cref{sec:greedy algorithms} on $\Col$ and returns the solution with the best unique coverage.

\begin{theorem}
    \label{thm:poly time data stream}
Let $\Stream$ denote a data stream of $\NumSets$ sets, $\Card \geq 2$ denote a cardinality constraint, $\MaxFreq \geq 2$ denote the maximum frequency in $\Stream$, $\MaxSize \geq 2$ denote the maximum set size in $\Stream$, and $\Error \in (0,1)$ denote an error parameter. Further, assume we have a polynomial-time algorithm \Alg{} with unique coverage ratio $\UniqRatio$ depending on $\Card, \MaxFreq,$ and $\MaxSize$. Then we can find a $(1/(2\UniqRatio) - 3\Error)$-approximation for \MUCprob{} with probability $1-1/\poly(\NumSets)$, using one pass, $\tilde{O}(\Card^2/\Error^{3})$ space, and in polynomial-time.
\end{theorem} 

\begin{proof}
    We use the subsampling approach from \cref{lemma:subsampling approach}.
    In each subsampled instance, we use an existing polynomial-time streaming algorithm \cite{mcgregor2019better} to find a $(1/2-\Error)$-approximation, $\Col$, for \MCprob{} in one pass while storing $\tilde{O}(\Card/\Error)$ sets; the sets in $\Col$ must be stored explicitly so that we can run \Alg{} on $\Col$. Running \Alg{} on $\Col$ returns a $\Sol\subseteq \Col$ that is a $(1/(2\UniqRatio)-\Error)$-approximation for the subsampled instance of \MUCprob{}. This implies a $(1/(2\UniqRatio) - 3\Error)$-approximation for the original instance, $\Stream$, with probability $1-1/\poly(\NumSets)$.
    The overall space complexity follows from explicitly storing $\tilde{O}(\Card / \Error)$ sets in each subsampled instance and from \cref{lemma:subsampling approach}, giving $ \lceil \log_{2}\UnivSize \rceil \cdot \tilde{O}(\Card / \Error) \cdot \tilde{O}(\Card / \Error^2) = \tilde{O}(\Card^2 / \Error^{3})$.
\end{proof}

\section{Algorithms for Bounding the Unique Coverage Ratio} \label{sec:greedy algorithms}

We here present algorithms that run in polynomial time. Given a collection,~$\Col$, each returns a subcollection $\Sol \subseteq \Col$ such that~$\Sol$'s unique coverage is within logarithmic ratio of $\Col$'s coverage. We hence call this the \emph{unique coverage ratio} of an algorithm. Our algorithms \UG{} (\cref{sec:unique greedy}), \UGF{} (\cref{sec:unique greedy freq}), and \UGS{} (\cref{sec:unique greedy size}) have unique coverage ratios that is logarithmic in $\ColSize = |\Col|$, $\MaxFreq$, and $\MaxSize$, respectively. 

\subsection{\UG{}} \label{sec:unique greedy}
We present and analyze our algorithm \UG{}, with pseudocode in \cref{alg:UG}. Its purpose is to take a collection $\Col$ of $\ColSize$ sets and return a collection $\Sol\subseteq\Col$ whose \emph{unique} coverage is at least a $1/\Har_{\ColSize}$ factor of $\Col$'s coverage. We formally state this in \cref{thm:UG ratio}.

\subparagraph*{\UG{} Overview.}
\UG{} first checks whether $\Col$'s unique coverage is at least $1/\Har_\ColSize$ of its own coverage. If so, then it immediately returns $\Col$ as the solution, which of course occurs if $\ColSize=1$. If not, then the idea is to discard the set $\SetSmall \in \Col$ with the smallest contribution to $\Col$'s unique coverage. It follows that the total loss in coverage from $\Col$ to $\Col\setminus\{\SetSmall\}$ is only $1 / \ColSize$ of $\Col$'s unique coverage. Observe that~$\SetSmall$ contributes at most $1 / \ColSize$ to $\Col$'s unique coverage, and any elements in $\SetSmall$ that are also in $\Col$'s non-unique cover must remain in $\RecCol$'s cover.
We then apply~$\UG$ recursively, to~$\RecCol$. As we show in \cref{thm:UG ratio}, since the performance of~$\UG$ relates unique coverage to coverage, $\RecCol$ has sufficient coverage so that the recursive solution from $\UG(\RecCol)$ has a unique coverage of at least $1/\Har_\ColSize$ of $\Col$'s coverage.

\begin{algorithm}
    \DontPrintSemicolon
    \caption{\UG}
    \label{alg:UG}
    \KwIn{$\Col$: collection of $\ColSize$ sets.}
    \KwOut{$\Sol\subseteq\Col$: subcollection satisfying $|\UniqCover(\Sol)| \geq {|\Cover(\Col)|}/{\Har_\ColSize} $.}
    \If
    {$|\UniqCover(\Col)| \geq {|\Cover(\Col)|}/{\Har_\ColSize}$ %
    }{\label{line:UG check ratio}
        $\Sol \gets \Col$\;
        \label{line:UG assign col}
    }
    \Else{\label{line:UG rec step}
        $\SetSmall \gets \argmin_{\SetS\in\Col} |\SetS\cap\UniqCover(\Col)|$\; \label{line:UG smallest unique set}
        $\Sol \gets \UG( \Col \setminus \{\SetSmall\})$\;
        \label{line:UG assign rec col}
    }
    \KwRet{$\Sol$}\; \label{line:UG return sol}
\end{algorithm}

\begin{theorem} \label{thm:UG ratio}
    Given a collection of $\ColSize$ sets, $\Col$, 
    \UG{} returns a collection $\Sol\subseteq\Col$ satisfying
    \begin{align}
        |\UniqCover(\Sol)| \geq \frac{|\Cover(\Col)|}{\Har_\ColSize} \,.
        \label{eqn:UG ratio}
    \end{align}
\end{theorem}

\begin{proof}
    We prove \cref{thm:UG ratio} by induction on $\ColSize = |\Col|$.

    \proofsubparagraph{Base Case.}
    If $\ColSize = 1$, then $|\UniqCover(\Col)| = |\Cover(\Col)|$ and $\Sol = \Col$, so we are done.

    \proofsubparagraph{Inductive Case.} Consider the case $\ColSize \geq 2$, and assume that \cref{thm:UG ratio} holds for $\ColSize - 1$. Then one of two subcases must hold: (i) $|\UniqCover(\Col)| \geq |\Cover(\Col)|/\Har_\ColSize$; or (ii) the negation,
        $|\UniqCover(\Col)| < |\Cover(\Col)| / \Har_\ColSize$.
    In subcase~(i), the \cref{line:UG check ratio} condition succeeds and $\UG$ returns the subcollection~$\Sol = \Col$, which clearly satisfies Ineq. \eqref{eqn:UG ratio}.

    So, we focus on subcase~(ii); since $|\UniqCover(\Col)| < |\Cover(\Col)| / \Har_\ColSize$,
     the \cref{line:UG check ratio} condition fails, thus \cref{line:UG assign rec col} assigns to $\Sol$ the solution from the recursive call on $\RecCol$.
     
  \cref{claim:coverage of rec col} lower bounds the coverage of this subcollection, $|\Cover(\RecCol)|$.
  Prior to that, we prove a handy claim.

  \begin{claim}
  \label{claim:one two}
  $|\Cover(\RecCol)| = |\NUniqCover(\Col)| + |\UniqCover(\Col)\setminus \SetSmall|$\,.
  \end{claim}
  \begin{claimproof}
  Observe that $\NUniqCover(\Col)$ and $\UniqCover(\Col)\setminus \SetSmall$ are disjoint; so it suffices to show that
  $\Cover(\RecCol) = \NUniqCover(\Col) \cup (\UniqCover(\Col)\setminus \SetSmall)$.
  We first show that RHS is a subset of LHS.
  Each element covered at least twice in $\Col$ remains covered in $\Col\setminus \{\SetSmall\}$; while each element uniquely covered in~$\Col$ that is not in $\SetSmall$ remains covered in $\Col\setminus \{\SetSmall\}$.
  Going the other way, consider an element that is in neither $\NUniqCover(\Col)$ nor
  $\UniqCover(\Col)\setminus \SetSmall$: then the only set it was in was~$T$, and hence it is not in~$\Col\setminus \{\SetSmall\}$.
  \end{claimproof}
    
    \begin{claim} \label{claim:coverage of rec col}
    Subcollection~$\Col \setminus \{\SetSmall\}$ satisfies
        \begin{align*}
            |\Cover(\RecCol)| \geq \left(1-\frac{1}{\ColSize \Har_\ColSize}\right) |\Cover(\Col)|\,.
        \end{align*}
    \end{claim}
    \begin{claimproof}
        First, observe that the contribution of each $\SetS\in\Col$ to $\UniqCover(\Col)$, i.e., $|\SetS\cap\UniqCover(\Col)|$, is disjoint from the contributions of all other sets in~$\Col$: each element in $\UniqCover(\Col)$ is covered by exactly one set.
        Therefore, the set $\SetSmall = \argmin_{\SetS\in\Col} |\SetS\cap\UniqCover(\Col)|$ in \cref{line:UG smallest unique set} satisfies
        \begin{align}
            |\SetSmall \cap \UniqCover(\Col)| \leq \frac{|\UniqCover(\Col)|}{\ColSize}\,. \label{eqn:smallest unique set}
        \end{align}
         With \cref{claim:one two}, we now prove \cref{claim:coverage of rec col}.
        \begin{align*}
            |\Cover(\RecCol)| &= |\NUniqCover(\Col)| + |\UniqCover(\Col)\setminus \SetSmall| \\
            &= |\Cover(\Col)| - |\UniqCover(\Col)| + |\UniqCover(\Col)| -  |\SetSmall\cap\UniqCover(\Col)|\\
            &= |\Cover(\Col)| -  |\SetSmall \cap \UniqCover(\Col)|\\
            &\geq |\Cover(\Col)| - \frac{|\UniqCover(\Col)|}{\ColSize} &\text{Ineq.~\eqref{eqn:smallest unique set}}\\
            &> |\Cover(\Col)| - \frac{|\Cover(\Col)|}{\ColSize\Har_\ColSize} &\text{subcase (ii)}\\
            &= \left(1-\frac{1}{\ColSize \Har_\ColSize}\right) |\Cover(\Col)|\,. \tag*{\claimqedhere}
        \end{align*}
    \end{claimproof}
    
    Recall that in subcase~(ii), \cref{line:UG assign rec col} assigns to~$\Sol$ the solution from the recursive call on~$\RecCol$. Since $|\RecCol| = \ColSize - 1$, we apply the inductive assumption to prove that~$\Sol$ satisfies Ineq.~\eqref{eqn:UG ratio}.
    \begin{align*}
        |\UniqCover(\Sol)| &\geq \frac{|\Cover(\RecCol)|}{\Har_{\ColSize-1}} &\text{inductive assumption} \\
        &\geq \frac{1}{\Har_{\ColSize-1}} \left(1-\frac{1}{\ColSize \Har_\ColSize}\right) |\Cover(\Col)| &\text{\cref{claim:coverage of rec col}} \\
        &= \frac{1}{\Har_{\ColSize-1}} \frac{\ColSize \Har_\ColSize - 1}{\ColSize \Har_\ColSize} |\Cover(\Col)| \\
        &= \frac{1}{\Har_{\ColSize-1}} \frac{\Har_\ColSize - \frac{1}{\ColSize}}{\Har_\ColSize} |\Cover(\Col)| \\
        &= \frac{|\Cover(\Col)|}{\Har_\ColSize}\,. &\text{$\Har_\ColSize - \frac{1}{\ColSize}=\Har_{\ColSize-1}$, for $\ColSize \geq 2$}
    \end{align*}
    
    We have proven that $\Sol$ satisfies Ineq.~\eqref{eqn:UG ratio} in the base case and the inductive case, proving \cref{thm:UG ratio}.
\end{proof}

\subsection{\UGF{}} \label{sec:unique greedy freq}
In this section, we present and analyze our algorithm \UGF{}, with pseudocode in \cref{alg:UGF}. The purpose of this algorithm is to take a collection $\Col$ with maximum frequency $\MaxFreqCol \leq |\Col|$, and an error parameter $\FreqError \in (0,1)$, and return a collection $\Sol \subseteq \Col$ whose unique coverage is at least a $(1/\Har_{\lceil \MaxFreqCol(\MaxFreqCol-1) / \FreqError \rceil} - \FreqError)$ factor of $\Col$'s coverage. By an appropriate choice of $\FreqError$ depending on $\MaxFreqCol$, this factor can be simplified to $1 / (2\ln \MaxFreqCol + o(\log \MaxFreqCol))$.

\subparagraph*{\UGF{} Overview.}
The idea of \UGF{} is to group all of the sets from $\Col$ into $\FreqColSize$ disjoint collections, $\Group_{1}, \dots, \Group_{\FreqColSize}$, so that the sets must be selected into the solution $\Sol$ in these groups, i.e., for each $\FreqColInd \in [\FreqColSize]$, either all of the sets in $\Group_{\FreqColInd}$, or none of the sets in $\Group_{\FreqColInd}$, must be selected into $\Sol$. Then, letting $\FreqCol$ be the collection of the covers of $\Group_{1}, \dots, \Group_{\FreqColSize}$, we can call \UG{} on $\FreqCol$ to find a selection of these covers, namely $\FreqSol$. The returned solution,~$\Sol$, is constructed by merging each $\Group_{\FreqColInd}$ whose cover was selected into $\FreqSol$, which ensures that the sets are selected in groups.

It can be seen that, by calling \UG{} on $\FreqCol$ and by \cref{thm:UG ratio}, the unique coverage of $\FreqSol$ is at least $1/\Har_{\FreqColSize}$ of $\FreqCol$'s coverage, and therefore at least $1/\Har_{\FreqColSize}$ of $\Col$'s coverage since $\FreqCol$ and $\Col$ have the same cover. The issue now is that sets from the same $\Group_{\FreqColInd}$ can overlap after being selected as a group into $\Sol$, which would make $\Sol$'s unique coverage smaller than $\FreqSol$'s unique coverage. This is addressed by setting the number of groups to be $\FreqColSize = \lceil \MaxFreqCol(\MaxFreqCol-1)/\FreqError \rceil$, and by the way \UGF{} allocates the sets into these groups: it allocates each $\SetS \in \Col$ to the group $\Group_{\FreqColInd}$ whose unique coverage intersects the least with $\SetS$. In this way, the total unique coverage that is lost due to overlapping sets in the same $\Group_{\FreqColInd}$ can be bounded by $\FreqError|\Cover(\Col)|$. Thus, $\Sol$'s unique coverage is at least $(1/\Har_{\FreqColSize} - \FreqError) = (1/\Har_{\lceil \MaxFreqCol(\MaxFreqCol-1)/\FreqError \rceil} - \FreqError)$ of $\Col$'s coverage. Details are given in the proof of \cref{thm:UGF ratio}.

\begin{algorithm}
    \DontPrintSemicolon
    \caption{\UGF}
    \label{alg:UGF}
    \KwIn{$\Col$: collection with maximum frequency $\MaxFreqCol \geq 2$, $\FreqError \in (0,1)$: error parameter.}
    \KwOut{$\Sol\subseteq\Col$: collection satisfying $|\UniqCover(\Sol)| \geq \left( 1/\Har_{\lceil \MaxFreqCol(\MaxFreqCol-1)/\FreqError \rceil} - \FreqError \right) |\Cover(\Col)|$.}
    $\FreqColSize \gets \left\lceil \MaxFreqCol(\MaxFreqCol-1) / \FreqError \right\rceil$\; \label{line:UGF num groups}
    \For(\tcp*[f]{Initialize empty groups}){$\FreqColInd \in [\FreqColSize]$}{\label{line:UGF init parts}
        $\Group_{\FreqColInd} \gets \varnothing$\;
    } 
    \For(\tcp*[f]{Allocate sets to groups}){$\SetS \in \Col$}{\label{line:UGF set loop}
        $\FreqColInd \gets \argmin_{\FreqColIndj \in [\FreqColSize]}{|\UniqCover(\Group_{\FreqColIndj}) \cap \SetS|}$\; \label{line:UGF find best part}
        $\Group_{\FreqColInd} \gets \Group_{\FreqColInd} \cup \{\SetS\}$\; \label{line:UGF update part}
    }
    $\FreqCol \gets \{ \Cover(\Group_{1}), \dots, \Cover(\Group_{\FreqColSize}) \}$ \tcp*[r]{Define collection of groups' covers} \label{line:UGF define freq col}
    $\FreqSol \gets \UG(\FreqCol)$\; \label{line:UGF call UG}
    $\Sol \gets \varnothing$\;
    \For(\tcp*[f]{Construct returned solution}){$\Cover(\Group_{\FreqColInd}) \in \FreqSol$}{\label{line:UGF merge groups loop}
        $\Sol \gets \Sol \cup \Group_{\FreqColInd}$\; \label{line:UGF update sol}
    }
    \KwRet{$\Sol$}\; \label{line:UGF return sol}
\end{algorithm}

\begin{theorem} \label{thm:UGF ratio}
    Given $\Col$ with  maximum frequency $\MaxFreqCol \geq 2$, $\FreqError \in (0,1)$ denote an error parameter, algorithm \UGF{} returns a collection $\Sol\subseteq\Col$ satisfying
    \begin{align}
        |\UniqCover(\Sol)| \geq \left( \frac{1}{\Har_{\lceil \MaxFreqCol(\MaxFreqCol-1) / \FreqError \rceil}} - \FreqError \right) |\Cover(\Col)|\,.
        \label{eqn:UGF ratio}
    \end{align}
    Moreover, setting $\FreqError = (9.28 \ln \MaxFreqCol)^{-1} (2 \ln \MaxFreqCol + 2 \ln \ln \MaxFreqCol + 5.61)^{-1}$, we obtain 
    \begin{align}
         |\UniqCover(\Sol)| \geq \left( \frac{1 - 1/ (9.28 \ln \MaxFreqCol)}{2 \ln \MaxFreqCol + 2 \ln \ln \MaxFreqCol + 5.61} \right) |\Cover(\Col)| \geq \frac{1}{2\ln \MaxFreqCol + o(\log \MaxFreqCol)} |\Cover(\Col)|\,. \label{eqn:UGF simple ratio}
    \end{align}
\end{theorem}

\begin{proof}
We first prove Ineq.~\eqref{eqn:UGF ratio}, starting with the following claim. 
    \begin{claim}
        $ \UniqCover(\Sol)
            = \UniqCover(\FreqSol) \setminus \bigcup_{\FreqColInd \colon \Cover(\Group_{\FreqColInd}) \in \FreqSol} \NUniqCover(\Group_{\FreqColInd})$.
        \label{claim:sol freqsol}
    \end{claim}
    
    \begin{claimproof}
    Consider an element~$x$ that is in exactly one set in~$\Sol$.
    This means that~$x$ is in exactly one set from exactly one group, say~$\Group_y$, chosen in~$\Sol$.
    Focusing on~$\FreqSol$, element~$x$ is clearly in~$\Cover(\Group_y)$ only, but might
    occur more than once in~$\Group_y$.
    Excluding elements that are in $\NUniqCover(\Group_{\FreqColInd})$ for every~$\FreqColInd$, we thus have the claim
    statement.
    \end{claimproof}

    \noindent
    With \cref{claim:sol freqsol}, we prove \cref{claim:lower bound UGF sol}.   
    \begin{claim} \label{claim:lower bound UGF sol}
    The solution $\Sol$ satisfies $|\UniqCover(\Sol)| \geq |\UniqCover(\FreqSol)| - \FreqError |\Cover(\Col)|$.
    \end{claim}
    \begin{claimproof}
     Given \cref{claim:sol freqsol}, $\Sol$ satisfies Ineq.~\eqref{eqn:init bound uniq cover sol},  
        \begin{align}
            |\UniqCover(\Sol)|
            &\geq |\UniqCover(\FreqSol)| - \sum_{\FreqColInd \colon \Cover(\Group_{\FreqColInd}) \in \FreqSol} |\NUniqCover(\Group_{\FreqColInd})| \notag \\
            &\geq |\UniqCover(\FreqSol)| - \sum_{\FreqColInd \in [\FreqColSize]} |\NUniqCover(\Group_{\FreqColInd})|\,. \label{eqn:init bound uniq cover sol}
        \end{align}
        
        We upper bound $\sum_{\FreqColInd \in [\FreqColSize]} |\NUniqCover(\Group_{\FreqColInd})|$ in Ineq.~\eqref{eqn:init bound uniq cover sol}. Let $\SetS_{\SetOrder}$ be the $\SetOrder^{\text{th}}$ set allocated in the \cref{line:UGF set loop} loop, let $\Group_{\FreqColInd, 0} = \varnothing$, and let $\Group_{\FreqColInd, \SetOrder}$ be the subcollection $\Group_{\FreqColInd}$ just after allocating $\SetS_{\SetOrder}$.
        
        Upon inserting $\SetS_{\SetOrder}$ into $\Group_{\FreqColInd}$, every element in $\UniqCover(\Group_{\FreqColInd, \SetOrder-1})$ that becomes non-uniquely covered is accounted for by $\UniqCover(\Group_{\FreqColInd, \SetOrder-1}) \cap \SetS_{\SetOrder}$. So it holds that
      $
            |\NUniqCover(\Group_{\FreqColInd, \SetOrder})| - |\NUniqCover(\Group_{\FreqColInd, \SetOrder-1})| = |\UniqCover(\Group_{\FreqColInd, \SetOrder-1}) \cap \SetS_{\SetOrder}|$.
        Thus, $|\NUniqCover(\Group_{\FreqColInd})|$ can be expressed by \cref{eqn:sum non uniq cover group} below, observing that for~$\SetS_{\SetOrder}$ the relevant difference term is zero.
        \begin{align}
            |\NUniqCover(\Group_{\FreqColInd})|
            &= \sum_{\SetS_{\SetOrder} \in \Group_{\FreqColInd}} ( |\NUniqCover(\Group_{\FreqColInd, \SetOrder})| - |\NUniqCover(\Group_{\FreqColInd, \SetOrder-1})| ) &\text{telescoping series} \notag \\
            &= \sum_{\SetS_{\SetOrder} \in \Group_{\FreqColInd}} |\UniqCover(\Group_{\FreqColInd, \SetOrder-1}) \cap \SetS_{\SetOrder}|\,. \label{eqn:sum non uniq cover group}
        \end{align}
        
        For each $\FreqColInd \in [\FreqColSize]$ and each $\SetS_{\SetOrder} \in \Group_{\FreqColInd}$, we want to show an upper bound of $|\UniqCover(\Group_{\FreqColInd, \SetOrder-1}) \cap \SetS_{\SetOrder}| \leq (\MaxFreqCol-1)|\SetS_{\SetOrder}| / \FreqColSize$. To see this, since the maximum frequency is $\MaxFreqCol$, each element $\Elem \in \SetS_{\SetOrder}$ is covered by at most $\MaxFreqCol - 1$ other sets, each possibly in a different group.
        Therefore, we have that
        \begin{align}
            \sum_{\FreqColIndj \in [\FreqColSize]} |\UniqCover(\Group_{\FreqColIndj, \SetOrder-1}) \cap \{\Elem\}| &\leq \MaxFreqCol - 1\,, \notag \\
            \sum_{\Elem \in \SetS_{\SetOrder}} \sum_{\FreqColIndj \in [\FreqColSize]} |\UniqCover(\Group_{\FreqColIndj, \SetOrder-1}) \cap \{\Elem\}| &\leq \sum_{\Elem \in \SetS_{\SetOrder}} (\MaxFreqCol - 1)\,, \notag \\
            \sum_{\FreqColIndj \in [\FreqColSize]} |\UniqCover(\Group_{\FreqColIndj, \SetOrder-1}) \cap \SetS_{\SetOrder}| &\leq (\MaxFreqCol - 1) |\SetS_{\SetOrder}|\,. \notag
        \end{align}
        Recall that $\SetS_{\SetOrder}$ was allocated to the group $\Group_{\FreqColInd} = \argmin_{\FreqColIndj \in [\FreqColSize]} |\UniqCover(\Group_{\FreqColIndj, \SetOrder-1}) \cap \SetS_{\SetOrder}|$ in \crefrange{line:UGF find best part}{line:UGF update part}. Therefore, by averaging on the above inequality, we have that for each $\FreqColInd \in [\FreqColSize]$ and each $\SetS_{\SetOrder}$ that ends up in $\Group_{\FreqColInd}$,
        \begin{align}
            |\UniqCover(\Group_{\FreqColInd, \SetOrder-1}) \cap \SetS_{\SetOrder}| \leq \frac{\MaxFreqCol-1}{\FreqColSize} |\SetS_{\SetOrder}|\,. \label{eqn:upper bound group inter}
        \end{align}
        Now we upper bound $\sum_{\FreqColInd \in [\FreqColSize]} |\NUniqCover(\Group_{\FreqColInd})|$.
        \begin{align}
            \sum_{\FreqColInd \in [\FreqColSize]} |\NUniqCover(\Group_{\FreqColInd})|
            &= \sum_{\FreqColInd \in [\FreqColSize]} \sum_{\SetS_{\SetOrder} \in \Group_{\FreqColInd}} |\UniqCover(\Group_{\FreqColInd, \SetOrder-1}) \cap \SetS_{\SetOrder}| &\text{\cref{eqn:sum non uniq cover group}} \notag \\
            &\leq \frac{\MaxFreqCol-1}{\FreqColSize} \sum_{\FreqColInd \in [\FreqColSize]} \sum_{\SetS_{\SetOrder} \in \Group_{\FreqColInd}} |\SetS_{\SetOrder}| &\text{Ineq.~\eqref{eqn:upper bound group inter}} \notag \\
            &= \frac{\MaxFreqCol-1}{\FreqColSize} \sum_{\SetS \in \Col} |\SetS| &\text{$\Group_{1}, \dots, \Group_{\FreqColSize}$ partitions $\Col$} \notag \\
            &\leq \frac{\MaxFreqCol-1}{\FreqColSize} \MaxFreqCol 
            |\Cover(\Col)| &\text{each $\Elem \in \Cover(\Col)$ is counted by $\leq\MaxFreqCol$ sets} \notag \\
            &= \frac{\MaxFreqCol(\MaxFreqCol-1)}{\lceil \MaxFreqCol (\MaxFreqCol-1) / \FreqError \rceil} |\Cover(\Col)| &\text{value of $\FreqColSize$:\ \cref{line:UGF num groups}} \notag \\
            &\leq \frac{\MaxFreqCol(\MaxFreqCol-1)}{ \MaxFreqCol (\MaxFreqCol-1) / \FreqError} |\Cover(\Col)| \notag \\
            &= \FreqError |\Cover(\Col)|\,. \label{eqn:upper bound sum overlap groups}
        \end{align}
        Applying the Ineq.~\eqref{eqn:upper bound sum overlap groups} upper bound to Ineq.~\eqref{eqn:init bound uniq cover sol} completes the proof of the claim.
    \end{claimproof}

    To prove Ineq.~\eqref{eqn:UGF ratio}, it remains to lower bound $|\UniqCover(\FreqSol)|$, in the inequality of \cref{claim:lower bound UGF sol}, in terms of $|\Cover(\Col)|$. Below, $|\Cover(\FreqCol)| = |\Cover(\Col)|$ holds since every $\SetS \in \Col$ is allocated to some $\Group_{\FreqColInd} \in \FreqCol$.
    \begin{align}
        |\UniqCover(\Sol)| &\geq |\UniqCover(\FreqSol)| - \FreqError |\Cover(\Col)| &\text{\cref{claim:lower bound UGF sol}} \notag \\
        &\geq \frac{|\Cover(\FreqCol)|}{\Har_{\FreqColSize}} - \FreqError |\Cover(\Col)|
        &\text{\cref{line:UGF call UG} and \cref{thm:UG ratio}} \notag \\
        &=\frac{|\Cover(\FreqCol)|}{\Har_{\lceil \MaxFreqCol(\MaxFreqCol-1) / \FreqError \rceil}} - \FreqError |\Cover(\Col)| &\text{value of $\FreqColSize$ (\cref{line:UGF num groups})} \notag \\
        &=\frac{|\Cover(\Col)|}{\Har_{\lceil \MaxFreqCol(\MaxFreqCol-1) / \FreqError \rceil}}  - \FreqError |\Cover(\Col)| &\text{$|\Cover(\FreqCol)| = |\Cover(\Col)|$}
        \notag \\
        &=\left( \frac{1}{\Har_{\lceil \MaxFreqCol(\MaxFreqCol-1) / \FreqError \rceil}}  - \FreqError \right) |\Cover(\Col)|\,. \notag
    \end{align}

    \proofsubparagraph{Ineq.~(\ref{eqn:UGF simple ratio}).}
        It remains to show that there exists a choice of $\FreqError$ that implies Ineq.~\eqref{eqn:UGF simple ratio}.
    \begin{claim}
        Setting $\FreqError = (9.28 \ln \MaxFreqCol)^{-1} (2 \ln \MaxFreqCol + 2 \ln \ln \MaxFreqCol + 5.61)^{-1}$ implies Ineq.~\eqref{eqn:UGF simple ratio}.
    \end{claim}
\begin{claimproof}
    First, we lower bound the ratio of Ineq.~\eqref{eqn:UGF ratio} without assigning $\FreqError$.
    \begin{align}
        \frac{1}{\Har_{\lceil \MaxFreqCol(\MaxFreqCol-1) / \FreqError \rceil}}  - \FreqError \notag
        &\geq \frac{1}{\ln ( \MaxFreqCol(\MaxFreqCol-1) / \FreqError + 1 ) + 1}  - \FreqError \notag \\
        &= \frac{1}{\ln ( \MaxFreqCol^2 / \FreqError - \MaxFreqCol / \FreqError + 1 ) + 1}  - \FreqError \notag \\
        &\geq \frac{1}{\ln ( \MaxFreqCol^2 / \FreqError ) + 1}  - \FreqError &\text{$\MaxFreqCol \geq 2$ and $\FreqError < 1$} \notag \\
        &= \frac{1}{2 \ln \MaxFreqCol + \ln(1/\FreqError) + 1}  - \FreqError\,. \label{eqn:UGF ratio lower bound}
    \end{align}

    Now, let $\ConstA, \ConstB > 0$ be constants to be assigned shortly, and assign $\FreqError = (\ConstA \ln \MaxFreqCol)^{-1} (2 \ln \MaxFreqCol + 2 \ln \ln \MaxFreqCol + \ConstB)^{-1}$.
    This choice of $\FreqError$ allows us to simplify the RHS of Ineq.~\eqref{eqn:UGF ratio lower bound} to $1/(2\ln \MaxFreqCol + o(\log \MaxFreqCol))$, though it does not strictly maximize it.\footnote{We did not maximize the RHS of Ineq.~\eqref{eqn:UGF ratio lower bound} at every $\MaxFreqCol \geq 2$ as there is no simple closed-form solution for $\FreqError$ in terms of $\MaxFreqCol$.} We continue below after substituting $\FreqError$ into Ineq.~\eqref{eqn:UGF ratio lower bound}. Note that the inequality below holds since: $\ln \ln \MaxFreqCol \leq \ln \MaxFreqCol / e$, and for all $\MaxFreqCol \geq 2 \colon 1 \leq \ln \MaxFreqCol / \ln 2$.
    \begin{align}
        &= \frac{1}{2 \ln \MaxFreqCol + \ln( (\ConstA \ln \MaxFreqCol) (2 \ln \MaxFreqCol + 2 \ln \ln \MaxFreqCol + \ConstB) ) + 1}  - \frac{1/ (\ConstA \ln \MaxFreqCol)}{2 \ln \MaxFreqCol + 2 \ln \ln \MaxFreqCol + \ConstB} \notag \\
        &\geq \frac{1}{2 \ln \MaxFreqCol + \ln( (\ConstA \ln \MaxFreqCol) (2 \ln \MaxFreqCol + 2 \ln \MaxFreqCol / e + \ConstB \ln \MaxFreqCol / \ln 2) ) + 1}  - \frac{1/ (\ConstA \ln \MaxFreqCol)}{2 \ln \MaxFreqCol + 2 \ln \ln \MaxFreqCol + \ConstB} \notag \\
        &= \frac{1}{2 \ln \MaxFreqCol + \ln( (\ConstA \ln^2 \MaxFreqCol) (2 + 2/e + \ConstB/\ln 2) ) + 1}  - \frac{1/ (\ConstA \ln \MaxFreqCol)}{2 \ln \MaxFreqCol + 2 \ln \ln \MaxFreqCol + \ConstB} \notag \\
        &= \frac{1}{2 \ln \MaxFreqCol + 2 \ln \ln \MaxFreqCol + \ln \ConstA + \ln (2 + 2/e + \ConstB/\ln 2) + 1}  - \frac{1/ (\ConstA \ln \MaxFreqCol)}{2 \ln \MaxFreqCol + 2 \ln \ln \MaxFreqCol + \ConstB}\,. \notag
    \end{align}
    To further simplify the above lower bound while ensuring it is always positive, we constrain $\ConstA$ such that, for all $\MaxFreqCol \geq 2\colon 1/ (\ConstA \ln \MaxFreqCol) < 1$, and constrain $\ConstB$ such that: $\ln \ConstA + \ln (2 + 2/e + \ConstB / \ln 2) + 1 \leq \ConstB$.
    \begin{align}
        &\geq \frac{1}{2 \ln \MaxFreqCol + 2 \ln \ln \MaxFreqCol + \ConstB}  - \frac{1/ (\ConstA \ln \MaxFreqCol)}{2 \ln \MaxFreqCol + 2 \ln \ln \MaxFreqCol + \ConstB} \notag \\
        &= \frac{1 - 1/ (\ConstA \ln \MaxFreqCol)}{2 \ln \MaxFreqCol + 2 \ln \ln \MaxFreqCol + \ConstB}\,. \label{eqn:UGF ratio lower bound constants}
    \end{align}
    We assign $\ConstA = 9.28$ and $\ConstB = 5.61$ as these values roughly maximize the RHS of Ineq.~\eqref{eqn:UGF ratio lower bound constants} at $\MaxFreqCol=2$ subject to the above constraints. In this way, we derive the ratio of Ineq.~\eqref{eqn:UGF simple ratio} and simplify it to $1/(2\ln \MaxFreqCol + o(\log \MaxFreqCol))$ below.
    \begin{align}
        &= \frac{1 - 1/ (9.28 \ln \MaxFreqCol)}{2 \ln \MaxFreqCol + 2 \ln \ln \MaxFreqCol + 5.61} \notag \\
        &= \frac{1}{(2 \ln \MaxFreqCol + 2 \ln \ln \MaxFreqCol + 5.61)(1 + 1/(9.28 \ln \MaxFreqCol - 1))} &\text{$1-\frac{1}{x} = \frac{1}{1+1/(x-1)}$} \notag \\
        &=\frac{1}{2 \ln \MaxFreqCol + o(\log \MaxFreqCol)}\,. \notag \tag*{\claimqedhere}
    \end{align}
\end{claimproof}
\noindent This completes the proof of \cref{thm:UGF ratio}.
\end{proof}

\subsection{\UGS{}} \label{sec:unique greedy size}
In this section, we present \UGS{}, with pseudocode in \cref{alg:UGS}, derived by combining \UGF{} with the approach in Theorem~4.2 of Demaine et al.~\cite{Demaine2008}. The purpose of this algorithm is to take a collection,~$\Col$, with maximum set size~$\MaxSizeCol$, an error parameter,~$\SizeError \in (0,1)$, and another error parameter,~$\SecSizeError \in (0,1)$, and return a $\Sol \subseteq \Col$ whose unique coverage is at least a factor of $\Col$'s coverage, where the factor depends on~$\MaxSizeCol$,~$\SizeError$, and~$\SecSizeError$. We state this formally in \cref{thm:UGS ratio} and give the proof for completeness; in fact, our proof slightly generalizes the proof of Theorem~4.2 by Demaine et al.~\cite{Demaine2008}, by allowing an arbitrary~$\SizeError$ rather than fixing $\SizeError = 1/2$.

\subparagraph*{\UGS{} Overview.}
\UGS{} first modifies $\Col$ into a `minimal' collection by discarding each set~$\SetT$ that uniquely covers no element. Then it checks if $\Col$'s size is at least an $\SizeError$ factor of its own coverage. If so, then it assigns $\Col$ to the solution $\Sol$. Otherwise, it constructs a sub-instance on those elements of frequency at most $\MaxSizeCol$ and calls $\UGF{}$ on the sub-instance with error $\SecSizeError$ to get~$\SizeSol$. Returned solution $\Sol$ comprises each set $\SetS \in \Col$ whose intersection with $\FreqdUniv$ was selected into $\SizeSol$.

\begin{algorithm}
    \DontPrintSemicolon
    \caption{\UGS}
    \label{alg:UGS}
    \KwIn{$\Col$: collection with maximum set size~$\MaxSizeCol$, $\SizeError \in (0,1)$: error parameter, $\SecSizeError \in (0,1)$: error parameter used in \UGF{}.}
    \KwOut{$\Sol\subseteq\Col$: subcollection satisfying $|\UniqCover(\Sol)| \geq \min( \SizeError, (1-\SizeError) \RatioUGF(\MaxSizeCol, \SecSizeError) ) |\Cover(\Col)| $ where $\RatioUGF(\MaxSizeCol, \SecSizeError) = {1}/\Har_{\lceil \MaxSizeCol (\MaxSizeCol-1) / \SecSizeError \rceil} - \SecSizeError$.}
    \While(\tcp*[f]{Make $\Col$ minimal}){\textnormal{$\SetT \gets \argmin_{\SetS \in \Col} |\SetS \cap \UniqCover(\Col)|$ satisfies $\SetT = \varnothing$}}{ \label{line:UGS reduce col loop}
        $\Col \gets \Col \setminus \{ \SetT \}$\; \label{line:UGS reduce col}
    }
    \If{$|\Col| \geq \SizeError |\Cover(\Col)|$}{ \label{line:UGS check size col}
        $\Sol \gets \Col$\; \label{line:UGS assign col to sol}
    }
    \Else(\tcp*[f]{Define instance on elements with freq.\,$\leq \MaxSize$}){
        $\FreqdUniv \gets \{ \Elem \in \Cover(\Col) : \Freq_{\Col}
        (\Elem) \leq \MaxSizeCol \}$\; \label{line:UGS define freq d elems}
        $\SizeCol \gets \{ \SetS \cap \FreqdUniv : \SetS \in \Col \}$\; \label{line:UGS define inter freq d}
        $\SizeSol \gets \UGF(\SizeCol, \SecSizeError)$\; \label{line:UGS call UGF}
        $\Sol \gets \varnothing$\; \label{line:UGS init sol}
        \For(\tcp*[f]{Construct returned solution}){$\SetS \cap \FreqdUniv \in \SizeSol$}{ \label{line:UGS add sets loop}
            $\Sol \gets \Sol \cup \{\SetS\}$ \; \label{line:UGS update sol}
        }
    }
    \KwRet{$\Sol$}\; \label{line:UGS return sol}
\end{algorithm}

\begin{theorem} \label{thm:UGS ratio}
Let $\Col$ denote a collection of sets, $\MaxSizeCol$ denote the maximum size of a set in $\Col$, $\SizeError \in (0,1)$ denote an error parameter, and $\SecSizeError \in (0,1)$ denote an error parameter passed to \UGF{}. Then \UGS{} returns a collection $\Sol\subseteq\Col$ satisfying
\begin{align}
    |\UniqCover(\Sol)| \geq \min ( \SizeError, (1-\SizeError)\RatioUGF(\MaxSizeCol,\SecSizeError) ) |\Cover(\Col)|\,, \label{eqn:UGS ratio}
\end{align}
where $\RatioUGF(\MaxSizeCol, \SecSizeError) = 1/(\Har_{\lceil \MaxSizeCol (\MaxSizeCol-1) / \SecSizeError \rceil}) - \SecSizeError$ denotes the (reciprocal) unique coverage ratio of \UGF{}.
Moreover, by assigning $\SizeError = (1/\RatioUGF(\MaxSizeCol,\SecSizeError)+1)^{-1}$, $\SecSizeError =  (\ConstA \ln \MaxSizeCol)^{-1} (2 \ln \MaxSizeCol + 2 \ln \ln \MaxSizeCol + \ConstB)^{-1}$, and appropriate constants to $\ConstA$ and $\ConstB$, we derive from Ineq.~\eqref{eqn:UGS ratio} the simpler inequality below.
\begin{align}
    |\UniqCover(\Sol)| \geq \frac{1}{2 \ln \MaxSizeCol + o(\log \MaxSizeCol)} |\Cover(\Col)|\,. \label{eqn:UGS simple ratio}
\end{align}
\end{theorem}

\begin{proof}
We begin by proving Ineq.~\eqref{eqn:UGS ratio}. Discarding sets from~$\Col$ that uniquely cover no elements, as in \crefrange{line:UGS reduce col loop}{line:UGS reduce col}, does not affect $\Cover(\Col)$.
   So assume that $\Col$ is minimal, i.e., every $\SetS \in \Col$ uniquely covers at least one element. This means that $|\UniqCover(\Col)| \geq |\Col|$.

    Now one of two cases must hold:~(i) $|\Col| \geq \SizeError|\Cover(\Col)|$; or~(ii) $|\Col| < \SizeError|\Cover(\Col)|$. The final ratio in Ineq.~\eqref{eqn:UGS ratio} is the the worst-case ratio.
    
    In case~(i), \UGS{} returns the solution $\Sol = \Col$, by the success of the condition in \cref{line:UGS check size col}. Further, $|\UniqCover(\Sol)| = |\UniqCover(\Col)| \geq |\Col| \geq \SizeError|\Cover(\Col)|$ holds by the minimality of $\Col$. Thus,~$\Sol$ satisfies Ineq.~\eqref{eqn:UGS ratio} in case~(i).

    In case~(ii), we show that the set $\FreqdUniv$ of elements $\Elem \in \Cover(\Col)$ with $\Freq_{\Col}(\Elem) \leq \MaxSizeCol$, as in \cref{line:UGS define freq d elems}, satisfies $|\FreqdUniv| \geq (1-\SizeError) |\Cover(\Col)|$. We have
    \begin{align}
        \MaxSizeCol\, |\Univ \setminus \FreqdUniv| &< \sum_{\SetS \in \Col} |\SetS| &\text{each $\Elem \in \Univ \setminus \FreqdUniv$ is counted by $> \MaxSizeCol$ sets} \notag \\
        &\leq \MaxSizeCol\, |\Col| &\text{max. set size is $\MaxSizeCol$} \notag \\
        &< \MaxSizeCol\, \SizeError |\Cover(\Col)| &\text{case~(ii)} \notag \\
        |\Univ \setminus \FreqdUniv| &< \SizeError |\Cover(\Col)| \notag \\
        |\FreqdUniv| &> (1-\SizeError) |\Cover(\Col)|\,. \notag
    \end{align}

By the \cref{line:UGS define inter freq d} definition, $\Cover(\SizeCol) = \FreqdUniv$,  so $|\Cover(\SizeCol)| \geq (1-\SizeError)|\Cover(\Col)|$. Therefore, calling \UGF{} on $\SizeCol$ with maximum frequency $\MaxSizeCol$ and error $\SecSizeError$, as in \cref{line:UGS call UGF}, gives a collection $\SizeSol$ satisfying $|\UniqCover(\SizeSol)| \geq \RatioUGF(\MaxSizeCol,\SecSizeError) |\Cover(\SizeCol)| \geq \RatioUGF(\MaxSizeCol,\SecSizeError) (1-\SizeError) |\Cover(\Col)|$.
Likewise, by definition, in \crefrange{line:UGS init sol}{line:UGS update sol}, $|\UniqCover(\Sol)| \geq |\UniqCover(\SizeSol)|$. Thus, $\Sol$ satisfies Ineq.~\eqref{eqn:UGS ratio} in case (2), proving \cref{thm:UGS ratio}.

\proofsubparagraph{Ineq.~(\ref{eqn:UGS simple ratio}).}
We first maximize $\min(\SizeError, (1-\SizeError)\RatioUGF(\MaxSizeCol, \SecSizeError))$ with respect to $\SizeError$ by setting the two arguments as equal; this makes the RHS of Ineq.~\eqref{eqn:UGS ratio} equal to $\SizeError = (1/\RatioUGF(\MaxSizeCol, \SecSizeError)+1)^{-1}$. Then, by assigning $\SecSizeError = (\ConstA \ln \MaxSizeCol)^{-1} (2 \ln \MaxSizeCol + 2 \ln \ln \MaxSizeCol + \ConstB)^{-1}$ and appropriate constants to $\ConstA$ and $\ConstB$, \UGF{}'s unique coverage ratio satisfies $\RatioUGF(\MaxSizeCol, \SecSizeError) \geq (2 \ln \MaxSizeCol + o(\log \MaxSizeCol))^{-1}$ as in \cref{thm:UGF ratio}. Further substituting this into the RHS of Ineq.~\eqref{eqn:UGS ratio} proves Ineq.~\eqref{eqn:UGS simple ratio}. This completes the proof of \cref{thm:UGS ratio}.
\end{proof}

\section{Space Lower Bound for a $(1.5 + o(1)) / (\ln \Card - 1)$-Approximation} \label{sec:lower bound}
In this section, we prove the following theorem:
\begin{theorem} \label{thm:lower bound adv}
Let $e^{2.5} \leq \Card \leq \NumSets$, $\SizeParam = \Card \log \NumSets + \log(\Card/0.05)$, and assume the universe size to be $\UnivSize = \Card(\Card-1)\sum_{\FreqInd=1}^{\Card}\left\lceil \SizeParam/\FreqInd \right\rceil $. Then every constant-pass randomized streaming algorithm
for \MUCprob{} that,
with probability at least~$0.95$,
has an approximation factor of $( 3/2 + 3 / \sqrt{2\Card} )/(\Har_\Card - 1)$,
requires~$\Omega(\NumSets/\Card^2)$ space.
\end{theorem}

\subsection{High-Level Ideas of the Reduction}
\label{sec:reduction-overview}
Similar to other approaches~\cite{mcgregor2019better,mcgregor2021maximum}, we prove our space lower bound by reducing the problem of $\Card$-player Set Disjointness (with the unique intersection promise) in the one-way communication model, denoted by \Disj{}, to \MUCprob{} in the stream model.

\subparagraph*{Set Disjointness in the One-Way Communication Model.}
In the one-way communication model, players must take turns in some fixed order to send a message to the player next in order, i.e., the $\PlayInd^{\text{th}}$ player can only send a message to the $(\PlayInd+1)^{\text{th}}$ player. There can be $\Pass \geq 1$ rounds of communication, where a single round is completed once every player has taken their turn. The last player can send a message back to the $1^{\text{st}}$ player at the end of a round if there is a next round.

In an instance of \Disj{}, each player $\PlayInd \in [\Card]$ is given a set of integers $\SetD_{\PlayInd}\subseteq[\NumSets]$. Moreover, it is promised that only two kinds of instances can occur:

\begin{description}
    \item[NO instance.] All sets $\SetD_{\PlayInd}$ are pairwise disjoint.
    \item [YES instance.] There is a unique integer $\IndStar\in[\NumSets]$ such that, for all $\PlayInd\in[\Card]$, $\IndStar\in\SetD_{\PlayInd}$.
\end{description}

The goal then is for the $\Card^{\text{th}}$ player (in the final round) to correctly answer, with probability at least~$0.9$,  whether the given sets form a YES or NO instance.

The communication complexity of \Disj{} in the $\Pass$-round one-way communication model is $\Omega(\NumSets / \Card)$, even for randomized protocols and even when the players can use public randomness~\cite{Chakrabarti2003}. As there are $\leq \Pass \Card$ messages, every (randomized) protocol for \Disj{}, at least one of the messages has size $\Omega(\NumSets / (\Pass \Card^2))$ in the worst case.

\subparagraph*{Reduction Overview.}
Given an instance of \Disj{}, the main goal of the reduction, with parameter~$\SizeParam$, is for the players to construct an instance of \MUCprob{} in a stream such that if they were given a NO instance of \Disj{}, the optimal unique coverage is less than $\SizeParam \Card^2 ( 1.5 + o(1) )$ (with high probability); whereas if the players were given a YES instance of \Disj{}, the optimal unique coverage is at least $\SizeParam\Card^2(\Har_\Card-1)$.
The ratio of these optimal unique coverages is less than 
$(1.5 + o(1))/(\Har_\Card - 1)$, so the players can use a $(1.5 + o(1))/(\Har_\Card - 1)$-approximation streaming algorithm on the \MUCprob{} instance to distinguish between a NO and YES instance.
By a standard argument, this implies a protocol for \Disj{} which involves each player sending the memory of the streaming algorithm in a message to the next player. A constant-pass $O(\Space)$-space streaming algorithm implies a protocol with a maximum message size of $O(\Space)$ in constant rounds of communication where each pass of the streaming algorithm takes one round. Thus, a $(1.5 + o(1))/(\Har_\Card - 1)$-approximation streaming algorithm for \MUCprob{} requires $\Omega(\NumSets / \Card^2)$ space.

\subparagraph*{Intuition of \MUCprob{} Construction.}
Here, we give the intuition for constructing the streaming instance of \MUCprob{} that achieves the optimal unique coverages above, with details in the proof of \cref{thm:lower bound adv}.

Let the universe of the \MUCprob{} instance be $\Univ = \Univ_{1} \cup \dots \cup \Univ_{\Card}$, where $\Univ_{1}, \dots, \Univ_{\Card}$ are $\Card$ disjoint sub-universes such that $|\Univ_{\FreqInd}| = \Card(\Card - 1)\lceil \SizeParam/\FreqInd \rceil$ (for a sufficiently large $\SizeParam$ as in \cref{thm:lower bound adv}). Then, for each $\StrInd \in [\NumSets]$, each player $\PlayInd$ constructs $\SetS_{\PlayInd}^{\StrInd} \subseteq \Univ$ such that $\SetS_{1}^{\StrInd}, \dots, \SetS_{\Card}^{\StrInd}$ satisfy the following properties:
    \begin{enumerate}
        \item Each set $\SetS_{\PlayInd}^{\StrInd}$  covers $\FreqInd/\Card$ proportion of $\Univ_\FreqInd$ for all $t$.  
        \item For each $\FreqInd \in [\Card]$, the sets $\SetS_{1}^{\StrInd}, \dots, \SetS_{\Card}^{\StrInd}$  partition a proportion, $\Prop_{\FreqInd} \in [0, 1]$, of $\Univ_\FreqInd$ while having a common intersection in the remaining $(1 - q_t)$ proportion of $\Univ_\FreqInd$. I.e., sets with identical~$i$ form a `sunflower', with their overlap concentrated in the sunflower's `kernel'.
        \item The choice of elements to be covered by $\SetS_{\PlayInd}^{\StrInd}$ are independent and uniform random with respect to $\StrInd \in [\NumSets]$.
    \end{enumerate}

The above construction ensures that (with high probability) every collection of $\NumSel \in [\Card]$ sets, $\SetS_{\PlayInd_1}^{\StrInd_1}, \dots, \SetS_{\PlayInd_\NumSel}^{\StrInd_\NumSel}$, with distinct $\StrInd_1, \dots, \StrInd_\NumSel$ has a unique coverage less than $\SizeParam \Card^2 ( 1.5 + o(1) )$ (with high probability); whereas a collection of $\NumSel = \Card$ sets with identical $\StrInd_1, \dots, \StrInd_\NumSel$ has a unique coverage of at least $\SizeParam\Card^2(\Har_\Card-1)$. Observe that $k \geq e^{2.5}$ ensures that $H_k - 1 > 1.5 + o(1)$.

Finally, to construct the streaming instance of \MUCprob{}, each player $\PlayInd$ inserts $\SetS_{\PlayInd}^{\StrInd}$ into the stream iff $\StrInd \in \SetD_{\PlayInd}$. This means that, given a NO instance, every set $\SetS_{\PlayInd}^{\StrInd}$ in the stream has a distinct $\StrInd$; whereas given a YES instance, there exists a collection of $\NumSel=\Card$ sets in the stream all indexed by $\IndStar$, the unique integer contained in all $\SetD_1, \dots, \SetD_\Card$. This results in the optimal unique coverages for the NO and YES instances as required.

\subsection{Proof of \cref{thm:lower bound adv}}
\label{sec:lb-proof}
    We show a reduction from \Disj{} to \MUCprob{}. Assume without loss of generality that the sets $\SetD_\PlayInd$ are padded so that $|\SetD_1 \cup \dots \cup \SetD_\Card | \geq \NumSets / 4 \geq \NumSets / \Card ^ 2$ holds for $\Card \geq 2$.

\subparagraph{Construction of \MUCprob{} Instance.}
    First, the players define the \MUCprob{} universe as $\Univ = \Univ_{1} \cup \dots \cup \Univ_{\Card}$, where $\Univ_{1}, \dots, \Univ_{\Card}$ are $\Card$ disjoint sub-universes such that $|\Univ_{\FreqInd}| = \Card(\Card - 1) \lceil \SizeParam/\FreqInd \rceil$. Observe that, as per the assumption in \cref{thm:lower bound adv}, we have $\UnivSize = |\Univ| = \sum_{\FreqInd=1}^{\Card}|\Univ_\FreqInd| = \Card(\Card - 1) \sum_{\FreqInd=1}^{\Card} \left\lceil \SizeParam/\FreqInd \right\rceil$. 

    The players now construct the \MUCprob{} sets so that they satisfy the properties given in the overview. For each $\StrInd \in [\NumSets]$ and $\FreqInd\in[\Card]$, the players define $\UniqUniv_{\FreqInd}^{\StrInd}\subseteq\Univ_{\FreqInd}$ as an independent and uniformly chosen random subset of size
    $\Prop_{\FreqInd} = (\Card-\FreqInd)/(\Card-1)$
    proportion of $\Univ_{\FreqInd}$;
    they then independently and uniformly-at-random partition $\UniqUniv_{\FreqInd}^{\StrInd}$ into $\Card$ equally sized sets, $\Part_{\FreqInd, 1}^{\StrInd}, \dots, \Part_{\FreqInd, \Card}^{\StrInd}$; the players agree on all of these choices using public randomness. For example, the players obtain a common random permutation of~$U_t$ and pick the corresponding parts in order.
    Note that $\UniqUniv_{\FreqInd}^{\StrInd}$ can be divided into $\Card$ equal sets since $|\UniqUniv_{\FreqInd}^{\StrInd}|/\Card$ is an integer, viz.
    \begin{align*}
        \frac{|\UniqUniv_{\FreqInd}^{\StrInd}|}{\Card}
        &= \frac{\Prop_{\FreqInd}|\Univ_{\FreqInd}|}{\Card} =
        \frac{(\Card-\FreqInd)\Card (\Card-1)}{\Card(\Card-1)} \left\lceil \frac{\SizeParam}{\FreqInd} \right\rceil
        = (\Card-\FreqInd)\left\lceil \frac{\SizeParam}{\FreqInd} \right\rceil.
    \end{align*}
    Then, for each $\StrInd \in [\NumSets]$, each player $\PlayInd$ defines their set $\SetS_{\PlayInd}^{\StrInd}$ such that, for each $\FreqInd \in [\Card]$, it covers the $\PlayInd^{\text{th}}$ set in the partition of $\UniqUniv_{\FreqInd}^{\StrInd}$, namely $\Part_{\FreqInd, \PlayInd}^{\StrInd}$; and it covers all of $\Univ_{\FreqInd} \setminus \UniqUniv_{\FreqInd}^{\StrInd}$. More precisely,
    \begin{equation*}
    \SetS_{\PlayInd}^{\StrInd} = \bigcup_{\FreqInd=1}^{\Card} \left[ \Part_{\FreqInd, \PlayInd}^{\StrInd} \cup (\Univ_{\FreqInd} \setminus \UniqUniv_{\FreqInd}^{\StrInd}) \right].
    \end{equation*}
    Observe \cref{claim:prop of univ t covered}, which we use in \cref{claim:exp uniq cover of dist i in t} later.
    \begin{claim} \label{claim:prop of univ t covered}
        For each $\StrInd \in [\NumSets]$, $\PlayInd \in [\Card]$, and $\FreqInd \in [\Card]$, $\SetS_{\PlayInd}^{\StrInd}$ covers $\FreqInd/\Card$ proportion of $\Univ_{\FreqInd}$.
    \end{claim}
    
    \begin{claimproof}
        The proportion of $\Univ_{\FreqInd}$ that $\SetS_{\PlayInd}^{\StrInd}$ covers is $|\SetS_{\PlayInd}^{\StrInd} \cap \Univ_{\FreqInd}|/|\Univ_{\FreqInd}|$, which we prove to be $\FreqInd / \Card$ below.
        \begin{align*}
            \frac{|\SetS_{\PlayInd}^{\StrInd} \cap \Univ_{\FreqInd}|}{|\Univ_{\FreqInd}|}
            &= \frac{|\Part_{\FreqInd,\PlayInd}^{\StrInd}|}{|\Univ_{\FreqInd}|} + \frac{|\Univ_{\FreqInd} \setminus \UniqUniv_{\FreqInd}^{\StrInd}|}{|\Univ_{\FreqInd}|}
            = \frac{|\UniqUniv_{\FreqInd}^{\StrInd}|}{\Card|\Univ_{\FreqInd}|} + \frac{|\Univ_{\FreqInd} \setminus \UniqUniv_{\FreqInd}^{\StrInd}|}{|\Univ_{\FreqInd}|}
            = \frac{\Prop_{\FreqInd}}{\Card} + 1 - \Prop_{\FreqInd} = \frac{\Card-\FreqInd}{\Card(\Card-1)} + 1 - \frac{\Card-\FreqInd}{\Card-1} \\
            &= \frac{\Card-\FreqInd}{\Card(\Card-1)} + \frac{\FreqInd-1}{\Card-1}
            = \frac{\Card-\FreqInd + \Card\FreqInd-\Card}{\Card(\Card-1)}
            = \frac{\Card\FreqInd - \FreqInd}{\Card(\Card-1)}
            = \frac{\FreqInd(\Card-1)}{\Card(\Card-1)}
            = \frac{\FreqInd}{\Card}\,.
            \tag*{\claimqedhere}
        \end{align*}
    \end{claimproof}
    
    To complete the construction, each player $\PlayInd$ inserts set $\SetS_{\PlayInd}^{\StrInd}$ into the stream iff $\StrInd \in \SetD_{\PlayInd}$. There are $\Theta(\NumSets)$ sets inserted into the stream since $\NumSets / 4 \leq |\SetD_1 \cup \dots \cup \SetD_\Card | \leq \NumSets $.

   \subparagraph{Upper Bound on Optimal Unique Coverage in a NO Instance.}
    Next, we prove \cref{lemma:uniq cover of dist i}, which implies the required upper bound on the optimal unique coverage in a NO instance. We say that a collection $\DistCol = \{ \SetS_{\PlayInd_{1}}^{\StrInd_{1}}, \dots, \SetS_{\PlayInd_{\NumSel}}^{\StrInd_{\NumSel}} \}$ with distinct $\StrInd_{1}, \dots, \StrInd_{\NumSel}$ is a \emph{player-distinct collection}; we also say that $\DistCol$ is \emph{feasible} if it contains at most $k$ sets. Note that in the \MUCprob{} instance generated from a NO instance of \Disj{}, every feasible solution is a player-distinct collection. Thus, it suffices to upper bound the unique coverage of every feasible player-distinct collection.
    \begin{lemma} \label{lemma:uniq cover of dist i}
        With probability at least~$0.95$, every feasible player-distinct collection $\DistCol$ satisfies
            $|\UniqCover(\DistCol)| <
            \SizeParam \Card^2
            ( {3}/{2} + {3}/\sqrt{2\Card} )$.
    \end{lemma}
    \begin{proof}   
    First, we upper bound  $\E[|\UniqCover(\DistCol) \cap \Univ_t|]$ for every feasible player-distinct collection,~$\DistCol$, and for every sub-universe~$U_t$ (\cref{claim:exp uniq cover of dist i in t}), then we use Hoeffding's inequality to prove an upper bound on $|\UniqCover(\DistCol) \cap \Univ_t|$ that with high probability, holds simultaneously for every~$\DistCol$ and~$U_t$ (\cref{claim:uniq cover of dist i in t}). Summing the bound in \cref{claim:uniq cover of dist i in t} over all~$\Card$ sub-universes suffices. 
    
    For a feasible player-distinct collection $\DistCol$, let $\RandVar_{\Elem, \DistCol}$ be the random variable such that $\RandVar_{\Elem, \DistCol} = 1$ if element $\Elem \in \UniqCover( \DistCol )$, and $\RandVar_{\Elem, \DistCol} = 0$ otherwise. This means that for each sub-universe $\Univ_t$, we have
    \begin{equation}
    |\UniqCover(\DistCol) \cap \Univ_\FreqInd| =  \sum_{\Elem \in \Univ_{\FreqInd}} \RandVar_{\Elem, \DistCol}; \text{ and so }
    |\UniqCover(\DistCol)| = \sum_{\FreqInd = 1}^{\Card} \sum_{\Elem \in \Univ_{\FreqInd}} \RandVar_{\Elem, \DistCol}\,.
    \label{eqn:indicators}
    \end{equation}
    
    \begin{claim} \label{claim:exp uniq cover of dist i in t}
        For each feasible player-distinct $\DistCol$ of $\NumSel \in [\Card]$ sets and each sub-universe $\Univ_\FreqInd$, it holds that
        $\E\left[|\UniqCover(\DistCol) \cap \Univ_t|\right]
            \leq \Card (\SizeParam + \FreqInd) \NumSel \left(1-{\FreqInd}/{\Card}\right)^{\NumSel-1}$.
    \end{claim}
    \begin{claimproof}

        Consider a sub-universe $\Univ_{\FreqInd}$. Each set $\SetS_{\PlayInd}^{\StrInd} \in \DistCol$ covers a $\FreqInd / \Card$ proportion of $\Univ_{\FreqInd}$ uniformly-at-random by \cref{claim:prop of univ t covered}, and the elements in $\SetS_{\PlayInd}^{\StrInd}$ are independent of the elements in the other sets in $\DistCol$ by the independent choices of $\UniqUniv_{\FreqInd}^{\StrInd}$ and $\Part_{\FreqInd, \PlayInd}^{\StrInd}$. So for each element $\Elem \in \Univ_{\FreqInd}$, we have the following probability:
        \begin{align}
            \Pr[\RandVar_{\Elem, \DistCol} = 1]
            &= \Pr\left[ \Elem \ \text{is uniquely covered by} \ \DistCol \right] \notag\\
            &= \NumSel \Pr\left[ \Elem  \ \text{is uniquely covered by a} \ \SetS_{\PlayInd}^{\StrInd}\in\DistCol \right] \notag\\
            &= \NumSel \frac{\FreqInd}{\Card} \left(1-\frac{\FreqInd}{\Card}\right)^{\NumSel-1}\,.
            \label{eqn:unique t k}
        \end{align}
        Thus, for each $\FreqInd\in[\Card]$, we upper bound $\E\left[|\UniqCover(\DistCol) \cap \Univ_t|\right] = \E[ \sum_{\Elem \in \Univ_{\FreqInd}} \RandVar_{\Elem, \DistCol} ]$ (recall~\cref{eqn:indicators}).
        \begin{align*}
            \E\left[ \sum_{\Elem \in \Univ_{\FreqInd}} \RandVar_{\Elem, \DistCol} \right]
            &= \sum_{\Elem \in \Univ_{\FreqInd}} \Pr\left[ \RandVar_{\Elem, \DistCol}=1 \right] \\
            &= |\Univ_{\FreqInd}| \frac{\NumSel \FreqInd}{\Card} \left(1-\frac{\FreqInd}{\Card}\right)^{\NumSel-1} & \text{\cref{eqn:unique t k}}\\
            &\leq \Card(\Card-1) \frac{\SizeParam + \FreqInd}{\FreqInd} \frac{\NumSel \FreqInd}{\Card} \left(1-\frac{\FreqInd}{\Card}\right)^{\NumSel-1} &\text{$|\Univ_{\FreqInd}| \leq \Card(\Card-1) \frac{\SizeParam + \FreqInd}{\FreqInd}$} \\
            &\leq \Card (\SizeParam + \FreqInd) \NumSel \left(1-\frac{\FreqInd}{\Card}\right)^{\NumSel-1}\,. \tag*{\claimqedhere}
        \end{align*}
        \end{claimproof}
        
        \begin{claim} \label{claim:uniq cover of dist i in t}
            With probability at least $0.95$, for every feasible player-distinct $\DistCol$ of $\NumSel \in [\Card]$ sets and every sub-universe $\Univ_\FreqInd$, it holds that
            \begin{align*}
                |\UniqCover(\DistCol) \cap \Univ_t|
                < \Card (\SizeParam + \FreqInd) \NumSel \left(1-\frac{\FreqInd}{\Card}\right)^{\NumSel-1} + \frac{\Card (\SizeParam+\FreqInd)}{(2\FreqInd)^{1/2}}\,.
            \end{align*}
          
        \end{claim}
        \begin{claimproof} 
            Recall from~\cref{eqn:indicators} that $|\UniqCover(\DistCol) \cap \Univ_t| = \sum_{\Elem \in \Univ_{\FreqInd}} \RandVar_{\Elem, \DistCol}$. Note that it is more convenient to upper bound $\sum_{\Elem\in\Univ_{\FreqInd}} \RandVar_{\Elem, \DistCol}$ separately for each $\FreqInd \in [\Card]$ rather than to upper bound $\sum_{\Elem\in\Univ} \RandVar_{\Elem, \DistCol}$ directly. This is because the random variables $\RandVar_{\Elem, \DistCol}$ are exchangeable over all $\Elem \in \Univ_{\FreqInd}$ but not over all $\Elem \in \Univ$, and exchangeability allows us to use Hoeffding's inequality on $\sum_{\Elem \in \Univ_{\FreqInd}}\RandVar_{\Elem, \DistCol}$ despite the $\RandVar_{\Elem, \DistCol}$ being dependent random variables (see \cite{Ramdas2023} for example). Let $\Mean_{\FreqInd,\NumSel} = \E[\sum_{\Elem \in \Univ_{\FreqInd}} \RandVar_{\Elem, \DistCol}]$ and $\Bound_{\FreqInd} = \Card (\SizeParam+\FreqInd)/(2\FreqInd)^{1/2}$. For each choice of $\DistCol$ and $\FreqInd \in [\Card]$, we upper bound $\Pr[\sum_{\Elem \in \Univ_{\FreqInd}} \RandVar_{\Elem, \DistCol} \geq \Mean_{\FreqInd,\NumSel} + \Bound_{\FreqInd}]$ using Hoeffding's inequality below.
            \begin{align*}
                \Pr\left[ \sum_{\Elem \in \Univ_{\FreqInd}} \RandVar_{\Elem, \DistCol} \geq \Mean_{\FreqInd,\NumSel} + \Bound_{\FreqInd} \right] 
                &\leq \exp\left(-\frac{2 \Bound_{\FreqInd}^2}{|\Univ_{\FreqInd}|}\right) \\
                &\leq \exp\left( -\frac{2 \Bound_{\FreqInd}^2 \FreqInd}{\Card^2(\SizeParam+\FreqInd)} \right) &\text{$|\Univ_{\FreqInd}| \leq \Card^2 \frac{(\SizeParam+\FreqInd)}{\FreqInd}$} \\
                &= \exp\left( -\frac{2 \Card^2 (\SizeParam+\FreqInd)^2 \FreqInd}{2 \Card^2 (\SizeParam+\FreqInd) \FreqInd} \right) &\text{$\Bound_{\FreqInd}=\frac{\Card (\SizeParam+\FreqInd)}{(2\FreqInd)^{1/2}}$} \\
                &\leq \exp (-\SizeParam) \\
                &= \exp \left( -\log \left(\frac{\NumSets^\Card \Card}{0.05} \right) \right) &\text{$\SizeParam = \Card \log \NumSets + \log\left(\frac{\Card}{0.05}\right)$} \\
                &= \frac{0.05}{\NumSets^{\Card} \Card}\,.
            \end{align*}
            We bound the probability that $\sum_{\Elem \in \Univ_{\FreqInd}} \RandVar_{\Elem, \DistCol} \geq \Mean_{\FreqInd,\NumSel} + \Bound_{\FreqInd}$ occurs for at least one choice of $\DistCol$ and $\FreqInd \in [\Card]$ by $0.05$ as follows: there are at most $\NumSets^{\Card}\Card$ choices $\DistCol$ and $\FreqInd\in[\Card]$ since there are at most $\sum_{\NumSel=1}^{\Card} \binom{\NumSets}{\NumSel} \leq \NumSets^{\Card}$ choices of $\NumSel\in[\Card]$ sets from a stream of length $\NumSets$, and there are $\Card$ choices of $\FreqInd$. Thus, by a union bound, the probability that $\sum_{\Elem \in \Univ_{\FreqInd}} \RandVar_{\Elem, \DistCol} \geq \Mean_{\FreqInd,\NumSel} + \Bound_{\FreqInd}$ occurs for at least one choice of $\DistCol$ and $\FreqInd \in [\Card]$ is at most $0.05$, so $\sum_{\Elem \in \Univ_{\FreqInd}} \RandVar_{\Elem, \DistCol} < \Mean_{\FreqInd,\NumSel} + \Bound_{\FreqInd}$ holds for every $\DistCol$ and every $\FreqInd \in [\Card]$ with probability at least $0.95$.
            
            Since $\Mean_{\FreqInd,\NumSel} = \E[\sum_{\Elem \in \Univ_{\FreqInd}} \RandVar_{\Elem, \DistCol}] \leq \Card (\SizeParam + \FreqInd) \NumSel \left(1-\FreqInd/\Card \right)^{\NumSel-1}$ by \cref{claim:exp uniq cover of dist i in t}, the required upper bound on $\sum_{\Elem \in \Univ_{\FreqInd}} \RandVar_{\Elem, \DistCol}$ holds with probability at least $0.95$, proving \cref{claim:uniq cover of dist i in t}.
        \end{claimproof}

        Finally, summing the inequality of \cref{claim:uniq cover of dist i in t} over the $k$ sub-universes gives an upper bound on $|\UniqCover( \DistCol )|$ that holds simultaneously for every feasible player-distinct collection $\DistCol$ with high probability.
        We finalize the proof of \cref{lemma:uniq cover of dist i}
        in \cref{claim:finalize lemma}.
        \begin{claim}
        With probability at least~$0.95$,
            $|\UniqCover(\DistCol)| <
            \SizeParam \Card^2
            \left( {3}/{2} + {3}/\sqrt{2\Card} \right)$.
            \label{claim:finalize lemma}
        \qedhere
        \end{claim}
        \begin{claimproof}
        \cref{claim:uniq cover of dist i in t} implies that with probability at least $0.95$,
        \begin{align}
            |\UniqCover( \DistCol )|
            &= \sum_{\FreqInd = 1}^{\Card} |\UniqCover(\DistCol) \cap \Univ_t|  \notag \\
            &< \sum_{\FreqInd = 1}^{\Card} \left( \Card (\SizeParam + \FreqInd) \NumSel \left(1-\frac{\FreqInd}{\Card}\right)^{\NumSel-1} + \frac{\Card (\SizeParam+\FreqInd)}{(2\FreqInd)^{1/2}} \right) &\text{\cref{claim:uniq cover of dist i in t}} \notag \\
            &= \Card \NumSel \sum_{\FreqInd = 1}^{\Card} (\SizeParam + \FreqInd) \left(\frac{\Card-\FreqInd}{\Card}\right)^{\NumSel-1} + \frac{\Card}{2^{1/2}} \sum_{\FreqInd = 1}^{\Card} \left( \frac{\SizeParam}{\FreqInd^{1/2}} + \FreqInd^{1/2} \right) \notag \\
            &\leq \Card \frac{\NumSel}{\Card^{\NumSel-1}} \sum_{\FreqInd = 1}^{\Card} (\SizeParam + \FreqInd) (\Card-\FreqInd)^{\NumSel-1} + \frac{\Card}{2^{1/2}} \sum_{\FreqInd = 1}^{\Card} \left( \frac{\SizeParam}{\FreqInd^{1/2}} + \Card^{1/2} \right) &\text{$\FreqInd \leq \Card$} \notag \\
            &= \Card \frac{\NumSel}{\Card^{\NumSel-1}} \sum_{\FreqIndA = 0}^{\Card-1} (\SizeParam + \Card-\FreqIndA) \FreqIndA^{\NumSel-1} + \frac{\Card}{2^{1/2}} \left(\sum_{\FreqInd = 1}^{\Card} \frac{\SizeParam}{\FreqInd^{1/2}} + \Card^{3/2} \right) &\text{let $\FreqIndA=\Card-\FreqInd$}  \notag \\
            &\leq \Card \frac{\NumSel}{\Card^{\NumSel-1}} \int_{0}^{\Card} \left( (\SizeParam+\Card)\FreqIndA^{\NumSel-1} - \FreqIndA^{\NumSel} \right) d\FreqIndA + \frac{\Card}{2^{1/2}} \left( \int_{0}^{\Card} \frac{\SizeParam}{\FreqInd^{1/2}} \, d\FreqInd + \Card^{3/2} \right)  \notag \\
            &= \Card \frac{\NumSel}{\Card^{\NumSel-1}} \left( \frac{(\SizeParam+\Card)\Card^{\NumSel}}{\NumSel} - \frac{\Card^{\NumSel+1}}{\NumSel+1} \right) + \frac{\Card}{2^{1/2}} ( 2\SizeParam \Card^{1/2} + \Card^{3/2} ) \notag \\
            &= \Card \frac{\Card^\NumSel}{\Card^{\NumSel-1}} \left( \frac{(\SizeParam+\Card)\NumSel}{\NumSel} - \frac{\Card\NumSel}{\NumSel+1} \right) + \frac{\Card^{3/2}}{2^{1/2}} ( 2\SizeParam + \Card ) \notag \\
            &= \Card^2 \left( \SizeParam+\Card - \frac{\Card\NumSel}{\NumSel+1} \right) + \frac{\Card^{3/2}}{2^{1/2}} ( 2\SizeParam + \Card ) \notag \\
            &= \Card^2 \left( \SizeParam + \frac{\Card(\NumSel + 1) - \Card\NumSel}{\NumSel+1} \right) + \frac{\Card^{3/2}}{2^{1/2}} ( 2\SizeParam + \Card ) \notag \\
            &= \Card^2 \left( \SizeParam+\frac{\Card}{\NumSel+1} \right) + \frac{\Card^{3/2}}{2^{1/2}} ( 2\SizeParam + \Card ) \notag \\
            &\leq \Card^2 \left( \SizeParam + \frac{\SizeParam}{\NumSel+1} \right)+ \frac{\Card^{3/2}}{2^{1/2}} 3\SizeParam
            &\text{$\Card \leq \SizeParam $} \notag \\
            &= \SizeParam \Card^2 \left( 1 + \frac{1}{\NumSel+1} + \frac{3}{(2 \Card)^{1/2}}  \right) \notag \\
            &\leq \SizeParam \Card^2 \left( \frac{3}{2} + \frac{3}{\sqrt{2\Card}} \right). &\text{%
            max.\ at $\NumSel=1$} \tag*{\claimqedhere}
        \end{align}
        \end{claimproof}
    \end{proof}

    \subparagraph{Lower Bound on Optimal Unique Coverage in a YES Instance.}
     \cref{lemma:uniq cover of iden i} supports the required lower bound on the optimal unique coverage in a YES instance. 
    \begin{lemma} \label{lemma:uniq cover of iden i}
        For all $\StrInd$,
     collection $\IdenCol = \{ \SetS_{1}^{\StrInd}, \dots, \SetS_{\Card}^{\StrInd} \}$ satisfies
            $|\UniqCover(\IdenCol)| \geq \SizeParam \Card^2 (\Har_\Card - 1).$
    \end{lemma}
    \begin{proof}
    For each $\FreqInd \in [\Card]$, $\IdenCol$ uniquely covers $|\UniqUniv_{\FreqInd}^{\StrInd}|$ by construction. Below, the inequality holds since $|\Univ_{\FreqInd}| = \Card(\Card - 1) \lceil \SizeParam/\FreqInd \rceil \geq \Card(\Card - 1) \SizeParam/\FreqInd$.
        \begin{align*}
            |\UniqCover(\IdenCol)|
            &= \sum_{\FreqInd=1}^{\Card} |\UniqUniv_{\FreqInd}^{\StrInd}|
            = \sum_{\FreqInd=1}^{\Card} \Prop_{\FreqInd}|\Univ_{\FreqInd}| 
            \geq \sum_{\FreqInd=1}^{\Card} \frac{\Card-\FreqInd}{\Card-1} \frac{\Card(\Card-1) \SizeParam}{\FreqInd}
            = \sum_{\FreqInd=1}^{\Card} \frac{\Card-\FreqInd}{\FreqInd} \SizeParam \Card \\
            &= \SizeParam \Card \sum_{\FreqInd=1}^{\Card} \left( \frac{\Card}{\FreqInd} - 1 \right) 
            = \SizeParam \Card \left( \Card \sum_{\FreqInd=1}^{\Card} \frac{1}{\FreqInd} - \Card \right)
            = \SizeParam \Card^2 \left( \Har_{\Card} - 1 \right). \tag*{\qedhere}
        \end{align*}
    \end{proof}

    To conclude,
    when the players reduce from a NO instance of \Disj{}, with probability at least~$0.95$, the optimal unique coverage is less than
    $\SizeParam \Card^2 ( 3/2 + 3 / \sqrt{2\Card} )$, since the streamed sets are player distinct and by \cref{lemma:uniq cover of dist i}; whereas when they reduce from a YES instance, the optimal unique coverage is at least $\SizeParam \Card^2 (\Har_{\Card}-1)$ since the sets $\SetS_{1}^{\IndStar}, \dots, \SetS_{\Card}^{\IndStar}$ are in the stream and by \cref{lemma:uniq cover of iden i}.
    The required optimal unique coverage in a NO instance fails with probability at most $0.05$.
    Let $\Ratio = ( 3/2 + 3 / \sqrt{2\Card} ) / (\Har_{\Card}-1)$.
    Given a randomized $O(\Space)$-space $\Ratio$-approximation streaming algorithm with failure probability at most $0.05$, the players can run this algorithm on the \MUCprob{} instance to distinguish between a NO or YES instance with failure probability at most $0.1$. This implies a protocol for \Disj{} with maximum message size~$O(\Space)$. Thus, a constant-pass randomized $\Ratio$-approximation streaming algorithm with success probability at least~$0.95$ requires $\Omega(\NumSets/\Card^2)$ space.

\section{Subsampling for the Data Stream} \label{sec:subsampling approach}
Here we outline the subsampling approach from \cite{mcgregor2021maximum}. Given a data stream instance of \MUCprob{}, it is possible to construct a number of \emph{subsampled} instances by sampling the universe $\Univ$ at varying rates. By running an algorithm on these subsampled instances in parallel, we lose only a small error in approximation w.h.p. while only needing to store sets of size $O(\Card \log \NumSets / \Error^2)$. We summarize the overall approach in \cref{lemma:subsampling approach} and give a proof sketch.

\begin{proof}[Proof Sketch of \cref{lemma:subsampling approach}.]
Given an instance of \MUCprob{} with universe $\Univ$ and subsets $\Stream$, let $\OptGuess$ be a guess of the optimal solution value; each subsampled instance  corresponds to some value of $\OptGuess$ (we calculate these guesses shortly). Let $\HashFunc : \Univ \rightarrow \{0, 1\}$ be a hash function that is $\Omega(\Card \log \NumSets / \Error^2)$-wise independent such that
\begin{align}
    \Pr[\HashFunc(\Elem) = 1] = \ProbHash = \frac{\Const \Card \log \NumSets}{\Error^2 \OptGuess}\,, \notag
\end{align}
where $\Const$ is a sufficiently large constant. Let $\UnivSample = \{ \Elem \in \Univ : \HashFunc(\Elem) = 1 \}$ be the subsampled universe, $\SetSSample = \SetS \cap \UnivSample$, $\StreamSample = \{ \SetSSample : \SetS \in \Stream \}$ be the subsampled sets, and $\OptValueSample$ be the optimal unique coverage in the subsampled instance. Further, let $\SolSample$ be a solution from $\StreamSample$ and $\Sol$ be the corresponding solution from the original collection $\Stream$. Then \cref{lemma:subsampling} below (a restatement of \cite[Lemma 23]{mcgregor2021maximum}) shows that, in a subsampled instance where $\OptGuess \leq \OptValue$,  w.h.p., the loss in approximation is at most $2 \Error$.
\begin{lemma}[\cite{mcgregor2021maximum}, Lemma 23] \label{lemma:subsampling}
    If $\OptGuess \leq \OptValue$, then with probability at least $1-1/\poly(\NumSets)$, we have that
    \begin{align}
         \ProbHash \OptValue (1+\Error) \geq \OptValueSample \geq \ProbHash \OptValue (1-\Error) \,. \notag
    \end{align}
    Furthermore, for some $\Ratio \in (0,1)$, if $\SolSample \subseteq \StreamSample$ satisfies $|\UniqCover(\SolSample)| \geq \Ratio (1-\Error) \ProbHash \OptValue$, then $|\UniqCover(\Sol)| \geq (\Ratio - 2\Error)\OptValue$.
\end{lemma}

We guess $\OptGuess = 2^{\Indi}$ for each $\Indi \in [\lceil\log_2 \UnivSize\rceil]$ and construct a subsampled instance for each $\OptGuess$ in parallel. Then, in the particular subsampled instance where $\OptValue / 2 \leq \OptGuess \leq \OptValue$, \cref{lemma:subsampling} implies the following upper bound on every set size $|\SetSSample|$ with probability $1-1/\poly(\NumSets)$.
\begin{align}
    |\SetSSample| \leq \OptValueSample \leq \ProbHash \OptValue (1+\Error) = \frac{\Const \Card \log \NumSets}{\Error^2 \OptGuess} \OptValue (1+\Error) \leq \frac{2\Const \Card \log \NumSets}{\Error^2} (1+\Error) = O\left( \frac{\Card \log\NumSets}{\Error^2} \right). \notag
\end{align}
To ensure that we only ever store sets of size $O( \Card \log\NumSets / \Error^2 )$, we terminate every subsampled instance that contains a set $\SetSSample$ with $|\SetSSample| > (2\Const \Card \log \NumSets/\Error^2) (1+\Error)$. This is safe to do since, w.h.p., we do not terminate the subsampled instance where $\OptValue/2 \leq \OptGuess \leq \OptValue$ by the above upper bound on $|\SetSSample|$ for every $\SetSSample$ in this particular instance.

This means that, out of the nonterminated subsampled instances, we should select the one with the smallest $\OptGuess$ and return the corresponding solution, giving an $(\Ratio-2\Error)$-approximation for the original instance w.h.p. (this works even if the smallest nonterminated guess satisfies $\OptGuess < \OptValue/2$ since \cref{lemma:subsampling} holds for all $\OptGuess \leq \OptValue$).

The overall space complexity, of $\lceil \log_2 \UnivSize \rceil \NumSetsStored \, O(\Card \log \NumSets \log \UnivSize/\Error^2)$, follows from the number of guesses of $\OptGuess$, and by the algorithm storing, for each guess at most $\NumSetsStored$ sets of size $O( \Card \log\NumSets/\Error^2 )$.
\end{proof}

\section{Conclusions} \label{sec:conclusions}
We are pleased to present a suite of algorithms, and a streaming lower bound, for \MUCprob{}.
The component algorithms that build a solution to \MUCprob{} from a solution \MCprob{} serve to support a fixed-parameter tractable approximation scheme (FPT-AS).
The lower bound shows that~$\Omega(\NumSets/\Card^2)$ space is required even to get within a~$(1.5 + o(1))/(\ln \Card - 1)$ factor of optimal.

A plasuible future direction would be to reduce, or indeed eliminate, the role of the upper bound of the unique coverage ratio,~$\UniqRatio$, in the kernel size in a FPT-AS.
This would match the kernel size used in existing FPT-ASs for \MCprob{}, but may not be possible due to the inherent hardness of \MUCprob{}. Another direction would be proving a streaming lower bound with a tighter approximation threshold. This may require a reduction from a different communication problem, rather than the renowned $\Card$-player Set Disjointess.

\newpage
\bibliography{main}

\begin{thebibliography}{10}

\bibitem{assadi2017tight}
Sepehr Assadi.
\newblock Tight space-approximation tradeoff for the multi-pass streaming set cover problem.
\newblock In {\em Proceedings of the 36th ACM SIGMOD-SIGACT-SIGAI Symposium on Principles of Database Systems}, pages 321--335, 2017.

\bibitem{ausiello2012online}
Giorgio Ausiello, Nicolas Boria, Aristotelis Giannakos, Giorgio Lucarelli, and V.\~Th Paschos.
\newblock Online maximum $k$-coverage.
\newblock {\em Discrete Applied Mathematics}, 160(13-14):1901--1913, 2012.

\bibitem{BateniEM17}
MohammadHossein Bateni, Hossein Esfandiari, and Vahab~S. Mirrokni.
\newblock Almost optimal streaming algorithms for coverage problems.
\newblock In Christian Scheideler and Mohammad~Taghi Hajiaghayi, editors, {\em Proceedings of the 29th {ACM} Symposium on Parallelism in Algorithms and Architectures, {SPAA} 2017, Washington DC, USA, July 24-26, 2017}, pages 13--23. {ACM}, 2017.
\newblock \href {https://doi.org/10.1145/3087556.3087585} {\path{doi:10.1145/3087556.3087585}}.

\bibitem{Bonnet2016}
{\'E}douard Bonnet, Vangelis~Th Paschos, and Florian Sikora.
\newblock Parameterized exact and approximation algorithms for maximum $ k $-set cover and related satisfiability problems.
\newblock {\em RAIRO-Theoretical Informatics and Applications-Informatique Th{\'e}orique et Applications}, 50(3):227--240, 2016.

\bibitem{Chakrabarti2003}
Amit Chakrabarti, Subhash Khot, and Xiaodong Sun.
\newblock Near-optimal lower bounds on the multi-party communication complexity of set disjointness.
\newblock In {\em 18th IEEE Annual Conference on Computational Complexity, 2003. Proceedings.}, pages 107--117. IEEE, 2003.

\bibitem{chitnis2019towards}
Rajesh Chitnis and Graham Cormode.
\newblock Towards a theory of parameterized streaming algorithms.
\newblock In {\em 14th International Symposium on Parameterized and Exact Computation}, 2019.

\bibitem{Chitnis2016}
Rajesh Chitnis, Graham Cormode, Hossein Esfandiari, MohammadTaghi Hajiaghayi, Andrew McGregor, Morteza Monemizadeh, and Sofya Vorotnikova.
\newblock Kernelization via sampling with applications to finding matchings and related problems in dynamic graph streams.
\newblock In {\em Proceedings of the twenty-seventh annual ACM-SIAM symposium on Discrete algorithms}, pages 1326--1344. SIAM, 2016.

\bibitem{Chitnis2015}
Rajesh Chitnis, Graham Cormode, Hossein Esfandiari, MohammadTaghi Hajiaghayi, and Morteza Monemizadeh.
\newblock New streaming algorithms for parameterized maximal matching \& beyond.
\newblock In {\em Proceedings of the 27th ACM symposium on Parallelism in Algorithms and Architectures}, pages 56--58, 2015.

\bibitem{Demaine2008}
Erik~D Demaine, Uriel Feige, MohammadTaghi Hajiaghayi, and Mohammad~R Salavatipour.
\newblock Combination can be hard: Approximability of the unique coverage problem.
\newblock {\em SIAM Journal on Computing}, 38(4):1464--1483, 2008.

\bibitem{DemaineIMV14}
Erik~D. Demaine, Piotr Indyk, Sepideh Mahabadi, and Ali Vakilian.
\newblock On streaming and communication complexity of the set cover problem.
\newblock In Fabian Kuhn, editor, {\em Distributed Computing - 28th International Symposium, {DISC} 2014, Austin, TX, USA, October 12-15, 2014. Proceedings}, volume 8784 of {\em Lecture Notes in Computer Science}, pages 484--498. Springer, 2014.
\newblock \href {https://doi.org/10.1007/978-3-662-45174-8\_33} {\path{doi:10.1007/978-3-662-45174-8\_33}}.

\bibitem{Guruswami2017}
Venkatesan Guruswami and Euiwoong Lee.
\newblock Nearly optimal {NP}-hardness of unique coverage.
\newblock {\em SIAM Journal on Computing}, 46(3):1018--1028, 2017.

\bibitem{Har-PeledIMV16}
Sariel Har{-}Peled, Piotr Indyk, Sepideh Mahabadi, and Ali Vakilian.
\newblock Towards tight bounds for the streaming set cover problem.
\newblock In {\em {PODS}}, pages 371--383. {ACM}, 2016.

\bibitem{Huang2022}
Chien-Chung Huang and Fran{\c{c}}ois Sellier.
\newblock Matroid-constrained maximum vertex cover: Approximate kernels and streaming algorithms.
\newblock In {\em SWAT 2022}, 2022.

\bibitem{Manurangsi2018}
Pasin Manurangsi.
\newblock A note on max $k$-vertex cover: Faster {FPT}-{AS}, smaller approximate kernel and improved approximation.
\newblock In {\em 2nd Symposium on Simplicity in Algorithms (SOSA 2019)}. Schloss Dagstuhl-Leibniz-Zentrum fuer Informatik, 2018.

\bibitem{mcgregor2021maximum}
Andrew McGregor, David Tench, and Hoa~T. Vu.
\newblock {Maximum Coverage in the Data Stream Model: Parameterized and Generalized}.
\newblock In {\em 24th International Conference on Database Theory}, 2021.

\bibitem{mcgregor2019better}
Andrew McGregor and Hoa~T. Vu.
\newblock Better streaming algorithms for the maximum coverage problem.
\newblock {\em Theory of Computing Systems}, 63(7):1595--1619, 2019.

\bibitem{Ramdas2023}
Aaditya Ramdas and Tudor Manole.
\newblock Randomized and exchangeable improvements of {M}arkov's, {C}hebyshev's and {C}hernoff's inequalities.
\newblock {\em arXiv preprint arXiv:2304.02611}, 2023.

\bibitem{saha2009maximum}
Barna Saha and Lise Getoor.
\newblock On maximum coverage in the streaming model \& application to multi-topic blog-watch.
\newblock In {\em Proceedings of the 2009 siam international conference on data mining}, pages 697--708. SIAM, 2009.

\bibitem{Sellier2023}
Fran{\c{c}}ois Sellier.
\newblock Parameterized matroid-constrained maximum coverage.
\newblock {\em arXiv preprint arXiv:2308.06520}, 2023.

\bibitem{SKOWRON201765}
Piotr Skowron.
\newblock {FPT approximation schemes for maximizing submodular functions}.
\newblock {\em Information and Computation}, 257:65--78, 2017.
\newblock URL: \url{https://www.sciencedirect.com/science/article/pii/S089054011730189X}, \href {https://doi.org/10.1016/j.ic.2017.10.002} {\path{doi:10.1016/j.ic.2017.10.002}}.

\bibitem{skowron2015fully}
Piotr Skowron and Piotr Faliszewski.
\newblock {Fully proportional representation with approval ballots: Approximating the MaxCover problem with bounded frequencies in FPT time}.
\newblock In {\em Proceedings of the AAAI Conference on Artificial Intelligence}, pages 2124--2130, 2015.

\bibitem{yu2013set}
Huiwen Yu and Dayu Yuan.
\newblock Set coverage problems in a one-pass data stream.
\newblock In {\em Proceedings of the 2013 SIAM international conference on data mining}, pages 758--766. SIAM, 2013.

\end{thebibliography}

\end{document}